\documentclass[reqno]{amsart}

\usepackage{booktabs} 
\usepackage{IEEEtrantools}
\usepackage{amssymb,latexsym,amsfonts,amsmath}
\usepackage{graphicx}
\usepackage{paralist}
\usepackage{capt-of}
\usepackage{xcolor}
\usepackage{dsfont}
\usepackage{tikz}
\usetikzlibrary{calc,shapes,arrows}
\usepackage{tikz}

\usepackage{tikz}
\usepackage{bm}
\def\baselinestretch{1}
\usepackage{dsfont}
\usepackage[bb=boondox]{mathalfa}
\usepackage{subfig}

\newtheorem{theorem}{Theorem}[section]
\newtheorem{lemma}[theorem]{Lemma}
\newtheorem{proposition}[theorem]{Proposition}

\newtheorem{definition}[theorem]{Definition}

\newtheorem{remark}[theorem]{Remark}
\newtheorem{assumption}[theorem]{Assumption}
\numberwithin{equation}{section}                            
 \newcommand{\II}{\mathcal{I}_d}

\newcommand{\R}{{\mathbb{R}}}

\newcommand{\N}{{\mathbb{N}}}

\newcommand{\eg}{{\it e.g.}}

\newcommand\norm[1]{\left\lVert#1\right\rVert}
\newcommand{\KK}{\mathcal{K}_{\infty}}

\newcommand{\Let}{:=}

\newcommand{\Z}{\mathbb{Z}}

\newcommand{\intcc}[1]{\ensuremath{{\left[#1\right]}}}

\begin{document}
	
	\title[Compositional Synthesis of Finite Abstractions for Networks of Systems:A Small-Gain Approach]	{Compositional Synthesis of Finite Abstractions for Networks of Systems: A Small-Gain Approach}

		\author{Abdalla Swikir$^1$}
	\author{Majid Zamani$^{2,3}$}
	\address{$^1$Hybrid Control Systems Group, Technical University of Munich, Germany.}
	\email{abdalla.swikir@tum.de}
	\address{$^2$Computer Science Department, University of Colorado Boulder, USA.}
	\address{$^3$Computer Science Department, Ludwig Maximilian University of Munich, Germany.}
	\email{majid.zamani@colorado.edu}
	\maketitle


\begin{abstract}                          
	In this paper, we introduce a compositional scheme for the construction of finite abstractions (a.k.a. symbolic models) of
interconnected discrete-time control systems. The compositional scheme is based on small-gain type reasoning. In particular, we use a notion of so-called alternating simulation functions as a relation between each subsystem and its symbolic model. Assuming some small-gain type conditions, we construct compositionally an overall alternating simulation function as a relation between an interconnection of symbolic models and that of original control subsystems. In such compositionality reasoning, the gains associated with the alternating simulation functions of the subsystems satisfy a certain ``small-gain" condition.
In addition, we introduce a technique to construct symbolic models together with their corresponding alternating simulation functions for discrete-time control subsystems under some stability property. 

Finally, we apply our results to the temperature regulation in a circular building by constructing compositionally a finite abstraction of a network containing $N$ rooms for any $N\geq3$. We use the constructed symbolic models as substitutes to synthesize controllers compositionally maintaining room temperatures in a comfort zone. We choose $N=1000$ for the sake of illustrating the results. We also apply our proposed techniques to a nonlinear example of fully connected network in which the compositionality condition still holds for any number of components.
In these case studies, we show the effectiveness of the proposed results  in comparison with the existing compositionality technique in the literature using a dissipativity-type reasoning.
\end{abstract}

In general, designing complex systems with respect to sophisticated control objectives is a challenging problem. 
In the past few years, several techniques have been developed to overcome those challenges.
One particular approach to address complex systems and control objectives is based on the construction of
finite abstractions (a.k.a. symbolic models) of the original control systems. Finite abstractions provide abstract descriptions of the continuous-space control systems in which each discrete state and input correspond to an aggregate of continuous states and inputs of the original system, respectively. 

In general, there exist two types of symbolic models: \emph{sound} ones whose behaviors (approximately) contain those of the concrete systems and \emph{complete} ones whose behaviors are (approximately) equivalent to those of the concrete systems \cite{Tabu}. Remark that existence of a complete symbolic model results in a sufficient and necessary guarantee in the sense that there exists a controller enforcing the desired specifications on the symbolic model \emph{if and only if} there exists a controller enforcing the same specifications on the
original control system. On the other hand, a sound symbolic model provides only a sufficient guarantee in the sense that failing to find a controller for the desired specifications on the symbolic model does not prevent the existence of a controller for the original control system. Since symbolic models are finite, controller synthesis problems can be algorithmically solved over them by resorting to automata-theoretic approaches \cite{MalerPnueliSifakis95,Thomas95}. 
Unfortunately, the construction of
symbolic models for large-scale interconnected systems is itself computationally a complex and challenging task.
An appropriate technique to overcome this challenge is to first construct symbolic models of the concrete subsystems individually and then establish a compositional framework using which one can construct abstractions of the overall network using those individual abstractions.

In the past few years, there have been several results on the compositional construction of finite abstractions of networks of control subsystems.
The framework introduced in \cite{Tazaki2008} based on the notion of interconnection-compatible approximate bisimulation relation provides networks of finite abstractions
approximating networks of stabilizable linear control systems. This work was extended in \cite{7403879} to networks of incrementally input-to-state stable nonlinear control
systems using the notion of approximate bisimulation relation. The recent result in \cite{Majumdar} introduces a new system relation, called (approximate) disturbance
bisimulation relation, as the basis for the compositional construction of symbolic models. Note that the proposed results in \cite{Tazaki2008,7403879,Majumdar} use the small-gain type conditions and provide \emph{complete} symbolic models of interconnected systems compositionally. The recent results in \cite{arxiv} introduce different conditions to handle the compositional construction of complete finite abstractions by leveraging techniques from dissipativity theory \cite{murat}. There are also other results in the literature \cite{meyer,omar,Kim} which provide sound symbolic models of interconnected systems, compositionally, without requiring any stability property or condition on the gains of subsystems.

In this work, we introduce a compositional approach for the construction of complete finite abstractions of interconnected nonlinear discrete-time control systems using more general small-gain type conditions. First, we introduce a notion of so-called alternating simulation functions inspired by Definition 1 in \cite{Girard20} as a system relation. Given alternating simulation functions between subsystems and their finite abstractions, we derive some small-gain type conditions to construct an overall alternating simulation function as a relation between the interconnected abstractions and the concrete network.
In addition, we provide a framework for the construction of finite abstractions together with their corresponding alternating simulation functions for discrete-time control systems satisfying incremental input-to-state stabilizability property \cite{angeli}. Finally, we illustrate our results by compositionally constructing finite abstractions of two networks of (linear and nonlinear) discrete-time control subsystems and their corresponding alternating simulation functions. These case studies particularly elucidate the effectiveness of the proposed results in comparison with the existing compositional result using dissipativity-type conditions in \cite{arxiv}. 

One can leverage the compositionally constructed finite abstractions here to synthesize controllers monolithically or also compositionally (see \cite[and references therein]{meyer}) to achieve some high-level properties. In particular, once finite abstractions are constructed for given concrete subsystems along with the corresponding alternating simulation functions, one can design local controllers also compositionally
	for those abstractions, and then refine them to the concrete subsystems provided that the given specification for the overall network is decomposable (see the first case study). Particularly, based on the assume-guarantee reasoning approach \cite{Rajamani}, the local controllers are synthesized by assuming that the other subsystems meet their local specifications. 

\textbf{ Related Work.} Results in \cite{Tazaki2008,7403879,Majumdar} use the small-gain type conditions (\cite[condition (17)]{Tazaki2008}, \cite[condition $r(A_k^{-1}C_k)<1$ in Theorem 1]{7403879}, and \cite[condition (22)]{Majumdar}) to facilitate the compositional construction of complete finite abstractions. Unfortunately, those small-gain type conditions are \emph{conservative}, in the sense that they are all formulated in terms of ``almost" linear gains, which means the considered subsystems should have a (nearly) linear behavior. Those conditions may not hold in general for systems with nonlinear gain functions (cf. Remark \ref{sgcv} in the paper). Here, we introduce more general small-gain type compositional conditions formulated in a general nonlinear form which can be applied to both linear and nonlinear gain functions without making any pre-assumptions on them. In addition, assuming a fully connected network, in the proposed compositionality results in \cite{Tazaki2008,7403879,Majumdar,arxiv} 
	the overall approximation error is either proportional to the summation of the approximation errors of finite abstractions of subsystems or lower bounded by the summation of positive and strictly increasing functions of quantization parameters of all subsystems. On the other hand, in the proposed results here the overall approximation error is proportional to the maximum of the approximation errors of finite abstractions of subsystems which are determined independently of the number of subsystems. Therefore, the results here can potentially provide complete finite abstractions for large-scale interconnected systems with much smaller approximation error in comparison with those proposed in \cite{Tazaki2008,7403879,Majumdar,arxiv} (cf. case studies for a comparison with \cite{arxiv}). 

\section{Notation and Preliminaries}\label{1:II}
\subsection{Notation}
We denote by $\R$, $\Z$, and $\N$ the set of real numbers, integers, and non-negative integers,  respectively.
These symbols are annotated with subscripts to restrict them in
the obvious way, \eg, $\R_{>0}$ denotes the positive real numbers. We denote the closed, open, and half-open intervals in $\R$ by $[a,b]$,
$(a,b)$, $[a,b)$, and $(a,b]$, respectively. For $a,b\in\N$ and $a\le b$, we
use $[a;b]$, $(a;b)$, $[a;b)$, and $(a;b]$ to
denote the corresponding intervals in $\N$.
Given $N\in\N_{\ge1}$, vectors $\nu_i\in\R^{n_i}$, $n_i\in\N_{\ge1}$, and $i\in[1;N]$, we
use $\nu=[\nu_1;\ldots;\nu_N]$ to denote the vector in $\R^n$ with
$n=\sum_i n_i$ consisting of the concatenation of vectors~$\nu_i$. Note that given any  $\nu\in\R^{n}$, $\nu \ge 0$ if $\nu_i \ge 0$ for any $i \in [1;n]$.
We denote the identity and zero matrix in $\R^{n\times n}$ by $I_n$ and $0_n$, respectively. 
The individual elements in a matrix $A\in \R^{m\times n}$, are denoted by $\{A\}_{ij}$, where  $i\in\intcc{1;m}$ and $j\in\intcc{1;n}$. We denote by $\norm{\cdot}$ and $\norm{\cdot}_2$ the infinity and Euclidean norm, respectively. Given any $a\in\R$, $\vert a\vert$ denotes the absolute value of $a$. Given sets $X$ and $Y$, we denote by $f:X\rightarrow Y$ an ordinary map of $X$ into $Y$, whereas $f:X\rightrightarrows Y$ denotes a set-valued map \cite{Rock0000}.
Given a function $f:\R^n\to \R^m$ and $\overline x\in\R^m$, we use $f\equiv \overline x$ to denote that $f(x)=\overline x$ for all $x\in\R^n$. If $\overline x$ is the zero vector, we simply write $f\equiv 0$. 
The identity map on a set $S$ is denoted by $1_{S}$. We denote by $|\cdot|$ the cardinality of a given set and by $\emptyset$ the empty set. 
{A set \mbox{$S\subseteq\R^n$} is a finite union of boxes if $S=\bigcup_{j=1}^MS_j$ for some $M\in\N$, where $S_j=\prod_{i=1}^n [c_i^j,d_i^j]\subseteq \R^n$ with $c^j_i<d^j_i$. For any set \mbox{$S\subseteq\R^n$} of the form of finite union of boxes, we define \mbox{$[S]_{\eta}=\{a\in S\,\,|\,\,a_{i}=k_{i}\eta,k_{i}\in\mathbb{Z},i=1,\ldots,n\}$}, where $0<\eta\leq\emph{span}(S)$, $\emph{span}(S)=\min_{j=1,\ldots,M}\eta_{S_j}$, \mbox{$\eta_{S_j}=\min\{|d_1^j-c_1^j|,\ldots,|d_n^j-c_n^j|\}$}. With a slight abuse of notation, we use $[S]_{0}:=S$. The set $[S]_{\eta}$ will be used as a finite approximation of the set $S$ with precision $\eta>0$. Note that $[S]_{\eta}\neq\varnothing$ for any $\eta\leq\emph{span}(S)$}. Given sets $S$ and $[S]_{\eta}$, $\vartheta_{\eta}:S\rightarrow [S]_{\eta}$ is an approximation map  that assigns for any $x\in S$ a representative point $\hat{x}\in [S]_{\eta}$ such that $\Vert x-\hat{x}\Vert<\eta$.
Given sets $U$ and $S\subset U$, the complement of $S$ with respect to $U$ is defined as $U\backslash S = \{x : x \in U, x \notin S\}.$
We use notations $\mathcal{K}$ and $\mathcal{K}_\infty$
to denote different classes of comparison functions, as follows:
$\mathcal{K}=\{\alpha:\mathbb{R}_{\geq 0} \rightarrow \mathbb{R}_{\geq 0} |$ $ \alpha$ is continuous, strictly increasing, and $\alpha(0)=0\}$; $\mathcal{K}_\infty=\{\alpha \in \mathcal{K} |$ $ \lim\limits_{r \rightarrow \infty} \alpha(r)=\infty\}$.
For $\alpha,\gamma \in \mathcal{K}_{\infty}$ we write $\alpha<\gamma$ if $\alpha(s)<\gamma(s)$ for all $s>0$, and $\mathcal{I}_d\in\mathcal{K}_{\infty}$ denotes the identity function.
\subsection{Discrete-Time Control Systems} 
In this paper we study discrete-time control systems of the following form.
\begin{definition}\label{def:sys1}
	A discrete-time control system $\Sigma$ is defined by the tuple	$\Sigma=(\mathbb X,\mathbb U,\mathbb W,\mathcal{U},\mathcal{W},f,\mathbb Y,h)$,
	where $\mathbb X, \mathbb U, \mathbb W,$ and $\mathbb Y$ are the state set, external input set, internal input set, and output set, respectively, and are assumed to be subsets of normed vector spaces with appropriate finite dimensions. Sets $\mathcal{U}$ and $\mathcal{W}$ denote the set
	of all bounded input functions $\nu:\N\rightarrow \mathbb U$ and $\omega:\N\rightarrow \mathbb W$, respectively. The set-valued map $f: \mathbb X\times \mathbb U \times \mathbb W\rightrightarrows \mathbb X $ is called the transition function, and $h:\mathbb X \rightarrow \mathbb Y$  is the output map.
	The discrete-time control system $\Sigma $ is described by difference inclusions of the form
	\begin{align}\label{eq:2}
	\Sigma:\left\{
	\begin{array}{rl}
	\mathbf{x}(k+1)\in& f(\mathbf{x}(k),\nu(k),\omega(k)),\\
	\mathbf{y}(k)=&h(\mathbf{x}(k)),
	\end{array}
	\right.
	\end{align}
	where $\mathbf{x}:\mathbb{N}\rightarrow \mathbb X $, $\mathbf{y}:\mathbb{N}\rightarrow \mathbb Y$, $\nu\in\mathcal{U}$, and $\omega\in\mathcal{W}$ are the state signal, output signal, external input signal, and internal input signal, respectively.
	
	System $\Sigma=(\mathbb X,\mathbb U,\mathbb W,\mathcal{U},\mathcal{W},f,\mathbb Y,h)$ is called deterministic if $|f(x,u,w)|\leq1$ $ \forall x\in \mathbb X, \forall u\in \mathbb U, \forall w \in \mathbb W$, and non-deterministic otherwise. System $\Sigma$ is called blocking if $\exists x\in \mathbb X, \forall u\in \mathbb U, \forall w \in \mathbb W $ where $|f(x,u,w)|=0$ and non-blocking if $|f(x,u,w)|\neq 0$ $ \forall x\in \mathbb X, \exists u\in \mathbb U, \exists w \in \mathbb W$.  System $\Sigma$ is called finite if $\mathbb X,\mathbb U,\mathbb W$ are finite sets and infinite otherwise. In this paper, we only deal with non-blocking systems.
\end{definition}	

Now, we introduce a notion of so-called alternating simulation functions, inspired by Definition 1 in \cite{Girard20}, which quantifies the error between systems $\Sigma$ and $\hat{\Sigma}$ both with internal inputs.
\begin{definition}\label{def:SFD1}
	Let $\Sigma=(\mathbb X,\mathbb U,\mathbb W,\mathcal{U},\mathcal{W},f,\mathbb Y,h)$ and $\hat{\Sigma}=(\hat{\mathbb{X}},\hat{\mathbb{U}},\hat{\mathbb{W}},\hat{\mathcal{U}},\hat{\mathcal{W}},\hat{f},\hat{\mathbb{Y}},\hat{h})$ where $\hat{\mathbb{W}}\subseteq{\mathbb{W}}$ and $\hat{\mathbb{Y}}\subseteq{\mathbb{Y}}$. A function $ V:\mathbb X\times \mathbb{\hat{X}} \to \mathbb{R}_{\geq0} $ is called an alternating simulation function from $\hat{\Sigma}$ to $\Sigma$ if $\forall x\in \mathbb X$ and $\forall \hat x\in\mathbb{\hat{X}}$,  one has
	\begin{align}\label{e:SFC11}
	\alpha (\Vert h(x)-\hat{h}(\hat{x})\Vert ) \leq V(x,\hat{x}),
	\end{align}
	and $\forall x\in \mathbb X $, $\forall \hat x\in \mathbb{\hat{X}}$, $\forall \hat u\in\mathbb{\hat{U}}$, $\exists u\in\mathbb U$,  $\forall w\in\mathbb W$, $\forall \hat w\in\mathbb{\hat{W}}$, $\forall x_{d} \in f(x,u,w), $ $\exists\hat{x}_{d} \in \hat{f}(\hat{x},\hat{u},\hat{w})$ such that one gets
	\begin{align}\label{e:SFC22}
	V&(x_{d},\hat{x}_{d})\\\notag
	&\leq \max\{\sigma(V(x,\hat{x})),\rho_{int}(\Vert w- \hat{w}\Vert ),\rho_{ext}(\Vert \hat{u}\Vert ),\varepsilon\},
	\end{align}
	for some $\alpha, \sigma,\rho_{int} \in \mathcal{K}_{\infty}$, where $\sigma<\mathcal{I}_d$, $\rho_{ext} \in \mathcal{K}_{\infty}\cup \{0\} $, and some $\varepsilon\in \mathbb{R}_{\geq 0}$.
\end{definition}
Let us point out some differences between our notion of alternating simulation
	function and the one in Definition~1 in~\cite{Girard20}. The notion of simulation
	function in \cite[Definition~1]{Girard20} is defined between two continuous-time control systems, whereas in Definition \ref{def:SFD1}, we define the alternating simulation
	function between two discrete-time control systems.    
	Moreover, there is no
	distinction between internal and external inputs in \cite[Definition~1]{Girard20}, whereas their distinctions in our work play a major role in providing the compositionality results later in the paper. Additionally, on the
	right-hand-side of \eqref{e:SFC22}, we introduce constant $\varepsilon\in \mathbb{R}_{\geq 0}$ to allow the relation to be defined between two (in)finite systems. The role of this constant will become clear in Section \ref{1:IV} where we introduce
	finite systems. Such a constant does not appear in \cite[Definition~1]{Girard20} which makes it only suitable for infinite systems.    
	Furthermore, we formulate the decay condition \eqref{e:SFC22} in a \emph{max-form}, while in \cite{Girard20} the decay condition is formulated in an \emph{implication-form}.

If $\Sigma$ does not have internal inputs, which is the case for interconnected systems (cf. Definition \ref{def:5}), Definition \ref{def:sys1} reduces to the tuple $\Sigma=(\mathbb X,\mathbb U,\mathcal{U},f,\mathbb Y,h)$ and the set-valued map $f$ becomes $f:\mathbb X\times\mathbb U\rightrightarrows\mathbb X$. Correspondingly, \eqref{eq:2} reduces to:
\begin{align}\label{eq:3}
\Sigma:\left\{
\begin{array}{rl}
\mathbf{x}(k+1)\in&f(\mathbf{x}(k),\nu(k)),\\
\mathbf{y}(k)=&h(\mathbf{x}(k)).
\end{array}\right.
\end{align}
Moreover, Definition \ref{def:SFD1} reduces to the following definition.
\begin{definition}\label{def:SFD2}
	Consider systems $\Sigma\!=(\mathbb X,\mathbb U,\mathcal{U},f,\mathbb Y,h)$ and  $\hat{\Sigma}\!=(\hat{\mathbb{X}},\hat{\mathbb{U}},\hat{\mathcal{U}},\hat{f},\hat{\mathbb{Y}},\hat{h}),$ where $\hat{\mathbb{Y}}\subseteq{\mathbb{Y}}$. A function $ \tilde{V}:\mathbb X\times \mathbb{\hat{X}} \to \mathbb{R}_{\geq0} $ is called an alternating simulation function from $\hat{\Sigma}$ to $\Sigma$ if $\forall x\in \mathbb X$ and $\forall \hat x\in\mathbb{\hat{X}}$, one has
	\begin{align}\label{e:SFC11a}
	\tilde{\alpha} (\Vert h(x)-\hat{h}(\hat{x})\Vert ) \leq \tilde{V}(x,\hat{x}),
	\end{align}
	and $\forall x\in \mathbb X $, $\forall \hat x\in \mathbb{\hat{X}}$, $\forall \hat u\in\mathbb{\hat{U}}$, $\exists u\in\mathbb U$, $\forall x_{d} \in f(x,u), $ $\exists\hat{x}_{d} \in \hat{f}(\hat{x},\hat{u})$ such that one gets
	\begin{align}\label{e:SFC22b}
	\tilde V(x_{d},\hat{x}_{d})\leq \max\{\tilde{\sigma}(\tilde{V}(x,\hat{x})),\tilde{\rho}_{ext}(\Vert \hat{u}\Vert ),\tilde{\varepsilon}\},
	\end{align}
	for some $\tilde{\alpha}, \tilde{\sigma} \in \mathcal{K}_{\infty}$, where $\tilde{\sigma}<\mathcal{I}_d$, $  \tilde{\rho}_{ext} \in \mathcal{K}_{\infty}\cup \{0\} $, and some $\tilde{\varepsilon}\in \mathbb{R}_{\geq 0}$.
\end{definition}
We say that a system $\hat{\Sigma} $ is approximately alternatingly simulated by a system $\Sigma $ or a system $\Sigma $ approximately alternatingly simulates a system $\hat{\Sigma} $, denoted by $\hat{\Sigma} \preceq _{\mathcal{AS}}  \Sigma$, if there
exists an alternating simulation function from $\hat{\Sigma} $ to $\Sigma $ as in Definition \ref{def:SFD2}.

We refer the interested readers to Section 3.2 in \cite{pt09} justifying in details the role of different quantifiers appeared before condition \eqref{e:SFC22b} in Definition \ref{def:SFD2} (condition \eqref{e:SFC22} in Definition \ref{def:SFD1}). In brief, those quantifiers capture the different role played by control inputs as well as nondeterminism in the system.

The next result shows that the existence of an alternating simulation function for systems without internal inputs implies the existence of an approximate alternating simulation relation between them as defined in \cite{Tabu}.	
\begin{proposition}
	Consider systems $\Sigma\!=(\mathbb X,\mathbb U,\mathcal{U},f,\mathbb Y,h)$ and  $\hat{\Sigma}\!=(\hat{\mathbb{X}},\hat{\mathbb{U}},\hat{\mathcal{U}},\hat{f},\hat{\mathbb{Y}},\hat{h}),$ where $\hat{\mathbb{Y}}\subseteq{\mathbb{Y}}$. Assume $\tilde V$ is an alternating simulation function from $\hat{\Sigma}$ to $\Sigma$ as in Definition \ref{def:SFD2} and that there exists $v\in \R_{>0}$ such that $\Vert \hat{u} \Vert \leq v$ $ \forall\hat{u} \in \mathbb{\hat{U}}$. Then, relation $R\subseteq\mathbb{X}\times \hat{\mathbb{X}}$ defined by $$R=\left\{(x,\hat{x})\in \mathbb{X}\times \hat{\mathbb{X}}|\tilde{V}(x,\hat{x})\leq \max\left\{\tilde{\rho}_{ext}(v),\tilde{\varepsilon}\right\}\right\}$$ is an $\hat\varepsilon$-approximate alternating simulation relation, defined in \cite{Tabu}, from $\hat{\Sigma}$ to $\Sigma$ with 
	\begin{align}\label{er}
	\hat\varepsilon=\tilde{\alpha}^{-1}(\max\{\tilde{\rho}_{ext}(v),\tilde{\varepsilon}\}).
	\end{align}
\end{proposition}
\begin{proof}{Proof.}
	The proof consists of showing that $(i)$ $\forall (x,\hat x)\in R$ we have $\Vert h(x)-\hat{h}(\hat{x})\Vert \leq \hat\varepsilon$; $(ii)$ $\forall (x,\hat x)\in R $ and $\forall \hat u\in\mathbb{\hat{U}}$, $\exists u\in\mathbb U$, such that $\forall x_{d} \in f(x,u),$ $\exists\hat{x}_{d} \in \hat{f}(\hat{x},\hat{u})$ satisfying $(x_d,\hat x_d)\in R$. The first item is a simple consequence of the definition of $R$ and condition \eqref{e:SFC11a} (i.e. $\tilde{\alpha} (\Vert h(x)-\hat{h}(\hat{x})\Vert ) \leq\tilde V(x,\hat{x})\leq\max\{\tilde{\rho}_{ext}(v),\tilde{\varepsilon}\}$), which results in $\Vert h(x)-\hat{h}(\hat{x})\Vert \leq \tilde{\alpha}^{-1}(\max\{\tilde{\rho}_{ext}(v),\tilde{\varepsilon}\}=\hat\varepsilon$. The second item follows immediately from the definition of $R$, condition \eqref{e:SFC22b}, and the fact that $\tilde{\sigma}<\mathcal{I}_d$. In particular, we have $\tilde V(x_{d},\hat{x}_{d})\leq\max\{\tilde{\rho}_{ext}(v),\tilde{\varepsilon}\}$ which implies $(x_d,\hat{x}_d)\in R$.
\end{proof}
\section{Compositionality Result}\label{1:III}
\label{s:inter}
In this section, we analyze networks of discrete-time control subsystems and drive a general small-gain type condition under which we can construct an alternating simulation function from a network of abstractions to the concrete network by using alternating simulation functions of the subsystems. The
definition of the network of discrete-time control subsystems is based on the notion of
interconnected systems described in~\cite{Tazaki2008}.
\subsection{Interconnected Control Systems}
We consider $N\in\N_{\ge1}$ original control subsystems $$\Sigma_i=(\mathbb X_i,\mathbb U_i,\mathbb W_i,\mathcal{U}_i,\mathcal{W}_i,f_i,\mathbb Y_i,h_i),$$ $i\in[1;N]$,
with partitioned internal inputs as
\begin{align}\label{eq:int1}
w_i&=[w_{i1};\ldots;w_{i(i-1)};w_{i(i+1)};\ldots;w_{iN}],\\
\mathbb W_i&=\prod_{j=1}^{N-1} \mathbb W_{ij},
\end{align}
with output map and set partitioned as
\begin{align}\label{eq:int2}
h_i(x_i)=&[h_{i1}(x_i);\ldots;h_{iN}(x_i)],\\\label{eq:int22}
\mathbb Y_i=&\prod_{j=1}^N \mathbb Y_{ij}.
\end{align}
We interpret the outputs $y_{ii}$ as external ones, whereas $y_{ij}$ with $i\neq j$ are internal ones which are used to define the
interconnected systems.
In particular, we assume that the dimension of vector $w_{ij}$ is equal
to that of vector $y_{ji}$.
If there is no connection from subsystem $\Sigma_{i}$ to
$\Sigma_j$, we set $h_{ij}\equiv 0$. Now, we define the notions of interconnections for control subsystems.
\begin{definition}\label{def:5}
		Consider $N\in\N_{\ge1}$ control subsystems 
		$\Sigma_i=(\mathbb X_i,\mathbb U_i,\mathbb W_i,\mathcal{U}_i,\mathcal{W}_i,f_i,\mathbb Y_i,h_i)$,
		$i\in[1;N]$, with the input-output structure given
		by $\eqref{eq:int1}-\eqref{eq:int22}$. The \emph{interconnected control
			system} $\Sigma=(\mathbb X,\mathbb U,\mathcal{U},f,\mathbb Y,h)$,
		denoted by
	$\mathcal{I}_{\mathcal M}(\Sigma_1,\ldots,\Sigma_N)$, where $\mathcal M \in\R^{N\times N}$ is a matrix with elements $\{\mathcal M\}_{ii}=0$, $\{\mathcal M\}_{ij}={\varpi_{ij}}$, $\forall i,j\in[1;N],i\neq j$, $0\leq{\varpi_{ij}}\leq\emph{span}(\mathbb{{Y}}_{ji})$, is defined by $ \mathbb X =\prod_{i=1}^N \mathbb X_i$,
		$ \mathbb U=\prod_{i=1}^N \mathbb U_i$, $\mathcal{U}=\prod_{i=1}^N\mathcal{U}_i$, $ \mathbb Y=\prod_{i=1}^N \mathbb Y_{ii}$, and maps
		\begin{align*}\notag
		f(x,u)&\!\Let\!\{\intcc{x_{d1};\ldots;x_{dN}}\,|\, \!x_{di}\in \!f_i(x_i,u_i,w_i)~ \forall i\!\in\![1;N]\},\\
		h(x)&\!\Let \!\intcc{h_{11}(x_1);\ldots;h_{NN}(x_N)},
		\end{align*}
		where $u=\intcc{u_{1};\ldots;u_{N}}$, $x=\intcc{x_{1};\ldots;x_{N}}$, and subject to the constraint:
		$$w_{ij}=\vartheta_{\varpi_{ij}}(y_{ji}),~[\mathbb{{Y}}_{ji}]_{\varpi_{ij}}\subseteq \mathbb{W}_{ij}, \forall i,j\in[1;N],i\neq j.$$
\end{definition}
In the above definition, whenever ${\varpi_{ij}}\neq 0$, the sets $\mathbb{{Y}}_{ji}$, $\forall i,j\in[1;N],~i\neq j$, are assumed to be finite unions of boxes.

An example of an interconnection of three control subsystems $\Sigma_1$, $\Sigma_2$,
	and $\Sigma_3$ is illustrated in Figure \ref{system1}.
\begin{figure}[ht]
	\begin{tikzpicture}[>=latex']
	\tikzstyle{block} = [draw, 
	thick,
	rectangle, 
	minimum height=1cm, 
	minimum width=1.5cm]
	
	\node at (-3.5,-0.75) {$\mathcal{I}_{0_3}(\Sigma_1,\Sigma_2,\Sigma_3)$};
	
	\draw[black,dashed] (-1.7,-3.8) rectangle (1.7,.7);
	
	\node[black,block] (S1) at (0,0) {$\Sigma_1$};
	\node[black,block] (S2) at (0,-1.5) {$\Sigma_2$};
	\node[black,block] (S3) at (0,-3.0) {$\Sigma_3$};
	
	\draw[black,->] ($(S1.east)+(0,0.25)$) -- node[very near end,above] {$y_{11}$} ($(S1.east)+(1.5,.25)$);
	\draw[black,<-] ($(S1.west)+(0,0.25)$) -- node[very near end,above] {$u_{1}$} ($(S1.west)+(-1.5,.25)$);
	
	\draw[black,->] ($(S2.east)+(0,-.0)$) -- node[very near end,below] {$y_{22}$} ($(S2.east)+(1.5,-.0)$);
	\draw[black,<-] ($(S2.west)+(0,-.0)$) -- node[very near end,below] {$u_{2}$} ($(S2.west)+(-1.5,-.0)$);
	
	\draw[black,->] ($(S3.east)+(0,-.25)$) -- node[very near end,below] {$y_{33}$} ($(S3.east)+(1.5,-.25)$);
	\draw[black,<-] ($(S3.west)+(0,-.25)$) -- node[very near end,below] {$u_{3}$} ($(S3.west)+(-1.5,-.25)$);
	
	\draw[black,->] 
	($(S1.east)+(0,-.43)$) -- node[very near end,above] {$y_{12}$} 
	($(S1.east)+(.5,-.43)$) --
	($(S1.east)+(.5,-.5)$) --
	($(S2.west)+(-.5,.5)$) --
	($(S2.west)+(-.5,.43)$) -- node[very near start,below] {$w_{21}$}
	($(S2.west)+(0,.43)$) ;
	
	\draw[black,->] 
	($(S2.east)+(0,.43)$) -- node[very near end,below] {$y_{21}$} 
	($(S2.east)+(.5,.43)$) --
	($(S2.east)+(.5,.5)$) --
	($(S1.west)+(-.5,-.5)$) --
	($(S1.west)+(-.5,-.43)$) -- node[very near start,above] {$w_{12}$}
	($(S1.west)+(0,-.43)$) ;
	
	\draw[black,->] 
	($(S2.east)+(0,-.43)$) -- node[very near end,above] {$y_{23}$} 
	($(S2.east)+(.5,-.43)$) --
	($(S2.east)+(.5,-.5)$) --
	($(S3.west)+(-.5,.5)$) --
	($(S3.west)+(-.5,.43)$) -- node[very near start,below] {$w_{32}$}
	($(S3.west)+(0,.43)$) ;
	\draw[black,->] 
	($(S3.east)+(0,.43)$) -- node[very near end,below] {$y_{32}$} 
	($(S3.east)+(.5,.43)$) --
	($(S3.east)+(.5,.5)$) --
	($(S2.west)+(-.5,-.5)$) --
	($(S2.west)+(-.5,-.43)$) -- node[very near start,above] {$w_{23}$}
	($(S2.west)+(0,-.43)$) ;
	
	\end{tikzpicture}
	\vspace{-0.1cm}
	\caption{Interconnection of three control subsystems $\Sigma_1$, $\Sigma_2$, and $\Sigma_2$ with $h_{13}=h_{31}=0$.}
	\label{system1}
	\vspace{-0.05cm}
\end{figure}
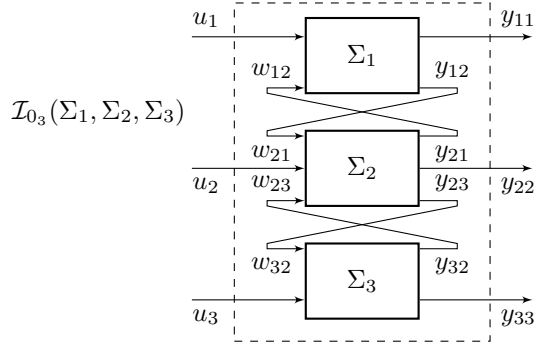

The following technical lemmas are used to prove some of the results in the next subsections. 
\begin{lemma}\label{lem1}
	For any $a,b \in \R_{>0}$, the following holds
	\begin{align}\label{fac1}
	a+b\leq\max\{(\mathcal{I}_d+\lambda)(a),(\mathcal{I}_d+\lambda^{-1})(b)\},
	\end{align}
	for any $\lambda \in \mathcal{K}_{\infty}$. 
\end{lemma}
\begin{proof}{Proof.}
	Define $c=\lambda^{-1}(b)$. Now, one has 
	\begin{align}\notag
	a+b=\!\left\{
	\begin{array}{lr}
	\!a+\lambda (c)\leq c+\lambda (c) =(\II+\lambda^{-1})(b)&~\text{if} ~a\leq c,\\
	\!a+\lambda (c)< a+\lambda (a)=(\II+\lambda)(a) &~\text{if} ~a> c, 
	\end{array}
	\right.
	\end{align}
	which implies \eqref{fac1}. 
\end{proof}
\vspace{-0.5cm}
The next lemma is borrowed from \cite{Kellett2014}.
\begin{lemma}\label{lem2}
	Consider $\alpha  \in \mathcal{K}$ and $\chi \in \mathcal{K}_{\infty}$, where $(\chi-\mathcal{I}_d)\in \mathcal{K}_{\infty}$. Then for any $a,b \in \R_{\geq0}$ 
	\begin{align*}
	\alpha(a+b)\leq\alpha\circ\chi(a)+\alpha\circ\chi\circ(\chi-\mathcal{I}_d)^{-1}(b).
	\end{align*}
\end{lemma}
Next subsection provides one of the main results of the paper on the compositional construction of abstractions for networks of systems. 
\subsection{Compositional Construction of Abstractions}
In this subsection, we assume that we are given $N$ original control subsystems $\Sigma=(\mathbb X_i,\mathbb U_i,\mathbb W_i,\mathcal{U}_i,\mathcal{W}_i,f_i,\mathbb Y_i,h_i)$ together with their corresponding abstractions
$\hat{\Sigma}_i\!=(\hat{\mathbb{X}}_i,\hat{\mathbb{U}}_i,\hat{\mathbb{W}}_i,\hat{\mathcal{U}}_i,\hat{\mathcal{W}}_i,\hat{f}_i,\hat{\mathbb{Y}}_i,\hat{h}_i)$ and alternating simulation functions $V_i$ from
$\hat\Sigma_i$ to $\Sigma_i$. Moreover, for functions $\sigma_i$, $\alpha_i$, and $\rho_{iint}$ associated with $V_i$, $\forall~ i\in [1;N]$, appeared in Definition \ref{def:SFD1}, we define
\begin{align}\label{gammad}
\gamma_{ii}\Let\sigma_{i}, ~\gamma_{ij}\Let(\mathcal{I}_d+\lambda)\circ\rho_{iint}\circ\chi\circ\alpha_{j}^{-1},
\end{align}
$\forall j \in [1;N],~j\neq i$, with arbitrarily chosen $\lambda,\chi\in\mathcal{K}_{\infty}$ with $(\chi-\mathcal{I}_d)\in \mathcal{K}_{\infty}$. Additionally, Let $\hat{\mathcal M} \in\R^{N\times N}$ be a matrix with elements $\{\hat{\mathcal M}\}_{ii}=0$, $\{\hat{\mathcal M}\}_{ij}={\hat{\varpi}_{ij}}$, $\forall i,j\in[1;N],i\neq j$, $0\leq{\hat{\varpi}_{ij}}\leq\emph{span}(\mathbb{{\hat{Y}}}_{ji})$.

The next theorem provides a compositional approach on the construction of abstractions of networks of control subsystems and that of the corresponding alternating simulation functions. 
\begin{theorem}\label{thm:3}
	Consider the interconnected control system
	$\Sigma=\mathcal{I}_{0_N}(\Sigma_1,\ldots,\Sigma_N)$ induced by
	$N\in\N_{\ge1}$
	control subsystems~$\Sigma_i$. Assume that each $\Sigma_i$ and its abstraction $\hat{\Sigma}_i$ admit an alternating simulation function $V_i$.
	Let the following holds:
	\begin{align}\label{SGC}
	\gamma_{i_1i_2}\circ\gamma_{i_2i_3}\circ\cdots\circ\gamma_{i_{r-1}i_r}\circ\gamma_{i_ri_1}<\mathcal{I}_d,
	\end{align}
	$\forall(i_1,\ldots,i_r)\in\{1,\ldots,N\}^r\backslash \left\{\{1\}^r,\ldots,\{N\}^r \right\}$, where $r\in \{2,\ldots,N\}$.
	Then, there exist $\delta_i \in \mathcal{K}_{\infty}$ such that 
	\begin{align}\notag
	\tilde{V}&(x,\hat{x})\Let\max\limits_{i}\{ \delta^{-1}_{i}\circ V_i(x_{i},\hat{x}_{i}) \} 
	\end{align}
	is an alternating simulation function from $\hat \Sigma={\mathcal{I}}_{\hat{\mathcal M}}(\hat{\Sigma}_1,\ldots,\hat{\Sigma}_N)$ to $\Sigma$. 
\end{theorem}
\begin{proof}{Proof.}
	Note that by using Theorem 5.2 in \cite{090746483}, condition \eqref{SGC} implies that $\exists~\delta_i \in \mathcal{K}_{\infty}$ $\forall i\in [1;N]$, satisfying 
	\begin{align}\label{gam}
	&\max\limits_{j\in [1;N]}\{\delta^{-1}_i\circ\gamma_{ij}\circ\delta_j\}<\mathcal{I}_d.
	\end{align} 
	Now, we show that $\eqref{e:SFC11a}$ holds for some $\mathcal{K}_{\infty}$ function $\tilde{\alpha}$. Consider any $x_i\in \mathbb X_i$, $\hat{x}_i \in \mathbb{\hat{X}}_i$, $\forall i\in[1;N]$. Then, one gets
	\begin{align*}\notag
	\Vert h(x)-\hat{h}(\hat{x})\Vert &=\max\limits_{i}\{\Vert h_{ii}(x_i)-\hat{h}_{ii}(\hat{x}_i)\Vert\}\\&\leq\max\limits_{i}\{\Vert h_{i}(x_i)-\hat{h}_{i}(\hat{x}_i)\Vert\}\\&\leq\max\limits_{i}\{\alpha^{-1}_i\circ V_i(x_{i},\hat{x}_{i})\}\\&\leq \hat{\alpha}\circ\max\limits_{i}\{ \delta^{-1}_{i}\circ V_i(x_{i},\hat{x}_{i})\},
	\end{align*}
	where $\hat{\alpha}(s)=\max\limits_{i}\{\alpha^{-1}_i\circ \delta_{i}(s)\}$ for all $s \in \mathbb{R}_{\geq0}$. By defining $\tilde{\alpha}= \hat{\alpha}^{-1}$, one obtains
	\begin{align*}\notag
	\tilde{\alpha}(\Vert h(x)-\hat{h}(\hat{x})\Vert )\leq \tilde{V}(x,\hat{x}),
	\end{align*}
	satisfying $\eqref{e:SFC11a}$.
	Now, we show that $\eqref{e:SFC22b}$ holds.
	Consider any 
	$x=\intcc{x_1;\ldots;x_N}\in\mathbb X$,
	$\hat x=\intcc{\hat x_1;\ldots;\hat x_N}\in\hat{\mathbb{X}}$, and any
	$\hat u=\intcc{\hat u_{1};\ldots;\hat u_{N}}\in\hat{\mathbb{U}}$. For any $i\in[1;N]$, there exists $u_i\in\mathbb U_i$, consequently, a vector $u=\intcc{u_{1};\ldots;u_{N}}\in\mathbb U$ such that  for any $x_d \in f(x,u)$ there exists  $\hat{x}_{d} \in \hat{f}(\hat{x},\hat{u})$ satisfying \eqref{e:SFC22} for each pair of subsystems $\Sigma_i$ and $\hat\Sigma_i$ with the internal inputs given by $\hat w_{ij}=\vartheta_{\hat{\varpi}_{ij}}(\hat{y}_{ji})$ $\forall i,j\in [1;N], j\neq i$. One gets the chain of inequalities in \eqref{comp}
	\begin{figure*}[h]
		\rule{\textwidth}{0.4pt}
		\begin{small}
			\begin{align}\notag
			\tilde{V}(x_{d},\hat{x}_{d})=&\max\limits_{i}\{\delta^{-1}_{i}\circ V_i(x_{d_i},\hat{x}_{d_i})\}\\\notag
			\leq& \max\limits_{i}\Big\{\delta^{-1}_{i}\big(\max\{\sigma_i\circ V_i(x_i,\hat{x}_i),\rho_{iint}(\Vert w_i- \hat{w}_i\Vert ),\rho_{iext}(\Vert \hat{u}_i\Vert ),\varepsilon_i\}\big)\Big\}\\\notag
			=& \max\limits_{i}\Big\{\delta^{-1}_{i}\big(\max\{\sigma_i\circ V_i(x_i,\hat{x}_i),\rho_{iint}(\max\limits_{j,j\neq i}\{\Vert w_{ij}- \hat{w}_{ij}\Vert \}),\rho_{iext}(\Vert \hat{u}_i\Vert ),\varepsilon_i\}\big)\Big\}\\\notag
			=& \max\limits_{i}\Big\{\delta^{-1}_{i}\big(\max\{\sigma_i\circ V_i(x_i,\hat{x}_i),\rho_{iint}(\max\limits_{j,j\neq i}\{\Vert y_{ji}-\vartheta_{\hat{\varpi}_{ij}(\hat{y}_{ji})}\Vert \}),\rho_{iext}(\Vert \hat{u}_i\Vert ),\varepsilon_i\}\big)\Big\}\\\notag
			=& \max\limits_{i}\Big\{\delta^{-1}_{i}\big(\max\{\sigma_i\circ V_i(x_i,\hat{x}_i),\rho_{iint}(\max\limits_{j,j\neq i}\{\Vert y_{ji}- \hat{y}_{ji}+\hat{y}_{ji}-\vartheta_{\hat{\varpi}_{ij}(\hat{y}_{ji})}\Vert \}),\rho_{iext}(\Vert \hat{u}_i\Vert ),\varepsilon_i\}\big)\Big\}\\\notag
			\leq&\max\limits_{i}\Big\{\delta^{-1}_{i}\big( \max\{\sigma_i\circ V_i(x_i,\hat{x}_i),\rho_{iint}(\max\limits_{j,j\neq i}\{\Vert h_{j}(x_j)\!\!-\!\!\hat{h}_{j}(\hat{x_j})\Vert\!\!+\!\hat{\varpi}_{ij} \}),\rho_{iext}(\Vert \hat{u}_i\Vert ),\varepsilon_i\}\big)\Big\}\\\notag
			\leq&\max\limits_{i}\Big\{\delta^{-1}_{i}\big( \max\{\sigma_i\circ V_i(x_i,\hat{x}_i),\rho_{iint}(\max\limits_{j,j\neq i}\{\alpha^{-1}_{j}\circ V_j(x_{j},\hat{x}_{j})\!+\!\hat{\varpi}_{ij} \}),\rho_{iext}(\Vert \hat{u}_i\Vert ),\varepsilon_i\}\big)\Big\}\\\notag
			\leq&\max\limits_{i}\Big\{\delta^{-1}_{i}\big( \max\{\sigma_i\circ V_i(x_i,\hat{x}_i),\rho_{iint}\circ\chi(\max\limits_{j,j\neq i}\{\alpha^{-1}_{j}\!\circ\! V_j(x_{j},\hat{x}_{j}) \})\!+\!\rho_{iint}\circ\!\chi\!\circ(\chi\!-\!\mathcal{I}_d)^{\!-\!1}(\max\limits_{j,j\neq i}\{\hat{\varpi}_{ij}\}),\rho_{iext}(\Vert \hat{u}_i\Vert ),\varepsilon_i\}\big)\!\Big\}\\\label{comp}
			\leq&\max\limits_{i}\Big\{\delta^{-1}_{i}\big( \max\{\sigma_i\circ V_i(x_i,\hat{x}_i),(\mathcal{I}_d+\lambda)\circ\rho_{iint}\circ\chi(\max\limits_{j,j\neq i}\{\alpha^{-1}_{j}\circ V_j(x_{j},\hat{x}_{j}) \}),\\\notag&\rho_{iext}(\Vert \hat{u}_i\Vert ),(\mathcal{I}_d+\lambda^{-1})\circ(\rho_{iint}\circ\chi\circ(\chi-\mathcal{I}_d)^{-1}(\max\limits_{j,j\neq i}\{\hat{\varpi}_{ij}\})+\varepsilon_i)\}\big)\Big\}\\\notag
			\leq&\max\limits_{i,j}\Big\{\delta^{-1}_{i}\big( \max\{\gamma_{ij}\circ V_j(x_j,\hat{x}_j),\rho_{iext}(\Vert \hat{u}_i\Vert ),\phi_i\}\big)\Big\}\\\notag
			=&\max\limits_{i,j}\Big\{\delta^{-1}_{i}\big( \max\{\gamma_{ij}\circ\delta_{j}\circ\delta^{-1}_{j}\circ V_j(x_j,\hat{x}_j),\rho_{iext}(\Vert \hat{u}_i\Vert ),\phi_i\}\big)\Big\}\\\notag
			\leq&\max\limits_{i,j,l}\Big\{\delta^{-1}_{i}\big( \max\{\gamma_{ij}\circ\delta_{j}\circ\delta^{-1}_{l}\circ V_l(x_l,\hat{x}_l),\rho_{iext}(\Vert \hat{u}_i\Vert ),\phi_i\}\big)\Big\}\\\notag
			=&\max\limits_{i,j}\Big\{\delta^{-1}_{i}\big( \max\{\gamma_{ij}\circ \delta_{j}\circ\tilde{V}(x,\hat{x}),\rho_{iext}(\Vert \hat{u}_i\Vert ),\phi_i\}\big)\Big\}\\\notag
			=&\max\Big\{\tilde{\sigma}\circ \tilde{V}(x,\hat{x}),\max\limits_{i}\big\{\delta^{-1}_{i}\circ\rho_{iext}(\Vert \hat{u}_i\Vert ),\delta^{-1}_{i}\big(\phi_i\big)\big\}\Big\},\notag
			\end{align}
		\end{small}
		\rule{\textwidth}{0.4pt}
	\end{figure*}
	for some arbitrarily chosen $\lambda,\chi\in\mathcal{K}_{\infty}$ with $(\chi-\mathcal{I}_d)\in \mathcal{K}_{\infty}$. Observe that, in inequalities \eqref{comp},  we used Lemma \ref{lem2} to go from line 7 to 8, and Lemma \ref{lem1} to go from line 8 to 9.  
	Define $\tilde{\sigma}$, $\tilde{\varepsilon}$, and $\tilde{\rho}_{ext}$ as follows:
	\begin{align*}\notag
	\tilde{\sigma}&\Let\max\limits_{i,j}\{\delta^{-1}_i\circ\gamma_{ij}\circ\delta_j\},\\
	\tilde{\varepsilon}&\Let \max\limits_{i}\{\delta^{-1}_{i}(\phi_i)\},\\
	\tilde{\rho}_{ext}(s)&\Let\left\{
	\begin{array}{lr}
	\max\limits_{i}\{\delta^{-1}_{i}\circ\rho_{iext}(s_i)\}\\
	s.t.~~~s=\Vert[s_1,\cdots,s_n]\Vert$, $s_i\geq0,
	\end{array}\right.
	\end{align*}
	where, $\forall i\in [1;N]$,
	\begin{align}\label{phi}
		\phi_i\!\!=\!\!(\mathcal{I}_d\!+\!\!\lambda^{\!-\!1}\!)\!\circ\!(\rho_{iint}\!\circ\!\chi\!\circ\!(\chi\!\!-\!\!\mathcal{I}_d)^{\!-\!1}\!(\max\limits_{j,j\neq i}\{{\hat{\varpi}_{ij}}\})\!+\!\varepsilon_i)\!\!\!
		\end{align} 
	Observe that it follows from \eqref{gam} that $\tilde{\sigma}<\mathcal{I}_d$. Then, one has
	\begin{align}\label{maxin}
	\tilde{V}(x_{d},\hat{x}_{d})\leq \max\{\tilde{\sigma}\circ \tilde{V}(x,\hat{x}),\tilde{\rho}_{ext}(\Vert \hat{u}\Vert ),\tilde{\varepsilon}\},
	\end{align}
	which satisfies $\eqref{e:SFC22b}$, and implies that $\tilde{V}$ is indeed an alternating simulation function from $\hat{\Sigma}$ to $\Sigma$. 
\end{proof}
Note that, similar technique was proposed in \cite{7496809} using nonlinear
	small-gain type condition to construct compositionally an approximate \emph{infinite} abstraction
	of an interconnected \emph{continuous-time} control system. Since in \cite{7496809} a simulation function between each subsystem and its abstraction is formulated in a \emph{dissipative-form} \cite{8000331}, an extra operator (the operator $D$ in \cite[equation (12)]{7496809}) is required to formulate the small-gain condition and to construct what is called an $\Omega$-path \cite[Definition 5.1]{090746483}, which is exactly the $\mathcal{K}_{\infty}$ functions $\delta_i,i\in N,$ that satisfy condition \eqref{gam} in our work. However, the definition of the simulation function in our work is formulated in a \emph{max-form} \cite{8000331} which results in not only simpler formulation of the small-gain condition but also the $\Omega$-path construction can be achieved without the need of the extra operator, see \cite[Section 8.4]{090746483}.
\begin{remark}
		Note that if, $\forall i\in [1;N]$, $\rho_{iint}$ are linear functions, i.e., $\rho_{iint}(a+b)=\rho_{iint}(a)+\rho_{iint}(b)$, $\forall a,b\in \R_{\ge0}$, we omit the $\mathcal{K}_{\infty}$ function $\chi$ in \eqref{gammad} and \eqref{phi}; hence, $\gamma_{ij}$ and $\phi_i$ in the previous theorem reduce to $\gamma_{ij}=(\mathcal{I}_d+\lambda)\circ\rho_{iint}\circ\alpha_{j}^{-1}$ and $\phi_i=(\mathcal{I}_d+\lambda^{-1})\circ(\rho_{iint}\circ(\max\limits_{j,j\neq i}\{\hat{\varpi}_{ij}\})+\varepsilon_i)$, $\forall i,j\in [1;N]$, $j\neq i$, respectively. Moreover, if $\hat{\varpi}_{ij}=0$,  we omit the $\mathcal{K}_{\infty}$ function $\lambda$ in \eqref{gammad} and \eqref{phi}. Therefore, $\gamma_{ij}$ and $\phi_i$ reduce to $\gamma_{ij}=\rho_{iint}\circ\alpha_{j}^{-1}$ and $\phi_i=\varepsilon_i$, $\forall i,j\in [1;N]$, $j\neq i$, respectively.
\end{remark}
\begin{remark}\label{sgcv} 
	We emphasize that the proposed small-gain type condition in \eqref{SGC} is much more general than the ones proposed in \cite{7403879,Majumdar}. To be more specific, consider the following system:
	\begin{align*}
	\Sigma:\left\{
	\begin{array}{rl}
	\mathbf{x}_1(k+1)&=a_1\mathbf{x}_1(k)+b_1\sqrt{\vert\mathbf{x}_2(k)\vert},\\
	\mathbf{x}_2(k+1)&=a_2\mathbf{x}_2(k)+b_2g(x_1(k)),
	\end{array}\right.
	\end{align*}
	where $0<a_1<1$, $0<a_2<1$, and function $g$ satisfies the following quadratic Lipschitz assumption: there exists an $L\in\R_{>0}$ such that: $\vert g(x)-g(x')\vert\leq L\vert x-x'\vert^2$ for all $x,x'\in \mathbb \R$. One can easily verify that functions $V_1(x_1,\hat x_1)=\vert x_1-\hat x_1\vert$ and $V_2(x_2,\hat x_2)=\vert x_2-\hat x_2\vert$ are alternating simulation functions from $\mathbf{x}_1$-subsystem to itself and $\mathbf{x}_2$-subsystem to itself, respectively.  
	Here, one can not come up with gain functions satisfying Assumption (A2) in \cite{7403879} globally (assumptions 1) and 2) in Theorem 3 in \cite{Majumdar} are continuous-time counterpart of Assumption (A2) in \cite{7403879}). In particular, those assumptions require existence of $\mathcal{K}_\infty$ functions being upper bounded by linear ones and lower bounded by quadratic ones which is impossible. On the other hand, the proposed small-gain condition \eqref{SGC} is still applicable here showing that $\tilde{V}(x,\hat{x})\Let\max\{ \delta^{-1}_{1}\circ V_1(x_{1},\hat{x}_{1}),\delta^{-1}_{2}\circ V_2(x_{2},\hat{x}_{2}) \}$ is an alternating simulation function from $\Sigma$ to itself, for some appropriate $\delta_{1},\delta_{2} \in\KK$ satisfying \eqref{gam} which is guaranteed to exist if $|b_1|\sqrt{|b_2|L}<1$ and $|b_2|(b_1L)^2<1$. 
\end{remark}

\begin{remark}
		Here, we provide a general guideline on the computation of $\mathcal{K}_{\infty}$ functions $\delta_i, i\in[1;N]$ as the following: $(i)$ In a general case of having $N\ge 1$ subsystems, functions $\delta_i, i\in[1;N]$, can be constructed numerically using the algorithm proposed in \cite{Eaves} and the technique provided in \cite[Proposition 8.8]{090746483}, see \cite[Chapter 4]{Rufferp}; $(ii)$ Simple construction techniques are provided in \cite{JIANG} and \cite[Section 9]{090746483} for the case of two and three subsystems, respectively; $(iii)$ the $\mathcal{K}_{\infty}$ functions $\delta_i, i\in[1;N]$, can be always chosen as identity functions provided that $\gamma_{ij}<\mathcal{I}_d$, $\forall~ i,j\in [1;N]$, for functions $\gamma_{ij}$ appeared in \eqref{gammad}.
\end{remark} 
\section{Construction of Symbolic Models}\label{1:IV}
In this section, we consider $\Sigma=(\mathbb X,\mathbb U,\mathbb W,\mathcal{U},\mathcal{W},f,\mathbb Y,h)$ as an infinite, deterministic control system and assume its output map $h$ satisfies the following general Lipschitz assumption: there exists an $\ell\in\KK$ such that: $\Vert h(x)-h(x')\Vert\leq \ell(\Vert x-x'\Vert)$ for all $x,x'\in \mathbb X$. Note that this assumption on $h$ is not restrictive at all provided that one is interested to work on a compact subset of $\mathbb X$. In addition, the existence of an alternating simulation function between $\Sigma$ and its finite abstraction is established under the assumption that $\Sigma$ is so-called incrementally input-to-state stabilizable as defined next.
\begin{definition}\label{ass:1}
	System $\Sigma\!=\!(\mathbb X,\mathbb U,\mathbb W\!,\mathcal{U},\mathcal{W}\!,f,\mathbb Y\!,h)$
	is called incrementally input-to-state stabilizable if there exist functions $\mathcal{H}:\mathbb X \to \mathbb U$ and  $ \mathcal{G}:\mathbb X \times \mathbb X \to \mathbb{R}_{\geq0} $  such that $\forall x,x'\in \mathbb X$, $\forall u,u'\in \mathbb U$, $\forall w,w' \in \mathbb W$, the inequalities:
	\begin{align}\label{eq:ISTFC1}
	\underline{\alpha} (\Vert x-x'\Vert ) \leq \mathcal{G}(x,x')\leq \overline{\alpha}(\Vert x-x'\Vert ),
	\end{align}
	and
	\vspace{-0.2cm}
	\begin{align}\notag
	\mathcal{G}&(f(x,\mathcal{H}(x)\!+\!u,w),f(x',\mathcal{H}(x')\!+\!u',w'))\!-\!\mathcal{G}(x,x')\\
	&\leq\!-\kappa(\mathcal{G}(x,x'))\!+\!\gamma_{int}(\Vert w\!-\! w'\Vert)\!+\!\gamma_{ext}(\Vert u\!-\! u'\Vert )\label{eq:ISTFC2}
	\end{align}
	hold for some $\underline{\alpha}, \overline{\alpha}, \kappa,\gamma_{int},\gamma_{ext} \in \mathcal{K}_{\infty}$.
\end{definition}

Remark that in Definition \ref{ass:1}, we implicitly assume that $\mathcal{H}(x)+u\in\mathbb{U}$ for any $x\in\mathbb{X}$ and any $u\in\mathbb{U}$. Note that any classically stabilizable linear control system is also incrementally stabilizable as in Definition \ref{ass:1}. For nonlinear control systems, the notion of incremental stabilizability as in Definition \ref{ass:1} is stronger than conventional stabilizability.
We refer the interested readers to \cite{ruffer} for detailed information on incremental input-to-state stability of discrete-time control systems.

Now, we construct a finite abstraction $\hat\Sigma$ of an incrementally input-to-state stabilizable control system $\Sigma$ as the following.
\begin{definition}\label{def:sym}
	Let $\Sigma=(\mathbb X,\mathbb U,\mathbb W,\mathcal{U},\mathcal{W},f,\mathbb Y,h)$ be incrementally input-to-state stabilizable as in Definition \ref{ass:1}, where $\mathbb X,\mathbb U,\mathbb W$ are assumed to be finite unions of boxes. One can construct a finite system 
	\begin{align}\label{eq:14}
	\hat{\Sigma}&=(\mathbb{\hat{X}},\mathbb{\hat{U}},\mathbb{\hat{W}},\hat{\mathcal{U}},\hat{\mathcal{W}},\hat{f},\hat{\mathbb{Y}}, \hat{h}),
	\end{align}
	\vspace{-0.2cm}
	where:
	\begin{itemize}
		\item $\mathbb{\hat{X}}=[\mathbb X]_\eta$, where $0<\eta\leq\emph{span}(\mathbb X)$ is the state set quantization parameter; 
		\item $\mathbb{\hat{U}}=[\mathbb U]_{\mu}$, where $0<\mu\leq\emph{span}(\mathbb U)$ is the external input set quantization parameter;
		\item $\mathbb{\hat{W}}=[\mathbb{W}]_{\hat{\varpi}}$, where $0\leq\hat{\varpi}\leq\emph{span}(\mathbb W)$ is the internal input set quantization parameter;
		\item $\hat{x}_{d}\in\hat{f}(\hat{x},\hat{u},\hat{w})$ iff $\Vert\hat{x}_{d}-f(\hat{x},\mathcal{H}(\hat{x})+\hat{u},\hat{w})\Vert \leq\eta$;
		\item $\hat{\mathbb{Y}}=\{h(\hat{x})\,\,|\,\,\hat{x} \in \mathbb{\hat{X}}\} $;
		\item $\hat{h}=h$.
	\end{itemize}
\end{definition}
Next, we establish the relation between $\Sigma$ and $\hat{\Sigma}$, introduced above, via the notion of alternating simulation function in Definition \ref{def:SFD1}. In particulate, we show that $\hat{\Sigma}$ is a complete finite abstraction of $\Sigma$ by proving that function $\mathcal{G}$ in Definition \ref{ass:1} is an alternating simulation function from $\hat{\Sigma}$ to ${\Sigma}$ and from $\Sigma$ to $\hat{\Sigma}$.

\begin{theorem}\label{thm:2}
	Let $\Sigma$ be an incrementally input-to-state stabilizable control system as in Definition \ref{ass:1} and $\hat{\Sigma}$ be a finite system as constructed in Definition \ref{def:sym}. Assume that there exists a function $\hat{\gamma}\in\mathcal{K}_{\infty}$ such that for any $x,x',x'' \in \mathbb{X}$ one has
	\begin{align}\label{eq:TI}
	\mathcal{G}(x,x')\leq \mathcal{G}(x,x'')+\hat{\gamma}(\Vert x'-x''\Vert)
	\end{align}
	for $\mathcal G$ as in Definition \ref{ass:1}.
	Then $\mathcal{G}$ is actually an alternating simulation function from $\hat{\Sigma}$ to $\Sigma$ and from $\Sigma$ to $\hat{\Sigma}$.
\end{theorem}

\begin{proof}{Proof.} Given the Lipschitz assumption on $h$ and since $\Sigma$ is incrementally input-to-state stabilizable, from \eqref{eq:ISTFC1}, $\forall x\in \mathbb{X}$ and $ \forall \hat{x} \in \mathbb{\hat{X}}
	$, we have 
	\begin{align}\notag
	\Vert h(x)-\hat{h}(\hat{x})\Vert\leq \ell(\Vert x-\hat{x}\Vert)\leq\hat{\alpha}(\mathcal{G}(x,\hat{x})),
	\end{align}	
	where $\hat{\alpha}=\ell\circ\underline{\alpha}^{-1}$. By defining $\alpha=\hat{\alpha}^{-1}$, one obtains
	\begin{align}\notag
	\alpha (\Vert h(x)-\hat{h}(\hat{x})\Vert )\leq \mathcal{G}(x,\hat{x}),
	\end{align}	
	satisfying \eqref{e:SFC11}.
	Now from \eqref{eq:TI}, $\forall x\in \mathbb{X}, \forall \hat{x} \in \mathbb{\hat{X}}, \forall \hat{u} \in \mathbb{\hat{U}},\forall w \in \mathbb{W},\forall \hat{w} \in \mathbb{\hat{W}}$, we have 
	\begin{align*}
	\mathcal{G}(f(x,&\mathcal{H}(x)+\hat{u},w),\hat{x}_{d})\\
	&\leq\mathcal{G}(f(x,\mathcal{H}(x)+\hat{u},w),f(\hat{x},\mathcal{H}(\hat{x})+\hat{u},\hat{w}))\\		&~~~+\hat{\gamma}(\Vert\hat{x}_{d}-f(\hat{x},\mathcal{H}(\hat{x})+\hat{u},\hat{w})\Vert)
	\end{align*}
	for any $\hat{x}_{d}\in\hat{f}(\hat{x},\hat{u},\hat{w})$.
	Now, from Definition \ref{def:sym}, the above inequality reduces to
	\begin{align*}
	\mathcal{G}&(f(x,\mathcal{H}(x)+\hat{u},w),\hat{x}_{d})\\
	&\leq\mathcal{G}(f(x,\mathcal{H}(x)+\hat{u},w),f(\hat{x},\mathcal{H}(\hat{x})+\hat{u},\hat{w}))+\hat{\gamma}(\eta).
	\end{align*}
	Note that by \eqref{eq:ISTFC2}, we get 
	\begin{align*}
	\mathcal{G}&(f(x,\mathcal{H}(x)+\hat{u},w),f(\hat{x},\mathcal{H}(\hat{x})+\hat{u},\hat{w}))-\mathcal{G}(x,\hat{x})\\
	&\leq-\kappa(\mathcal{G}(x,\hat{x}))+\gamma_{int}(\Vert w- \hat{w}\Vert ).
	\end{align*}	
	Hence, $\forall x\in \mathbb{X}, \forall \hat{x} \in \mathbb{\hat{X}}, \forall \hat{u} \in \mathbb{\hat{U}},$ and $\forall w \in \mathbb{W},\forall \hat{w} \in \mathbb{\hat{W}}
	$, one obtains
	\begin{align*}
	\mathcal{G}&(f(x,\mathcal{H}(x)+\hat{u},w),\hat{x}_{d})-\mathcal{G}(x,\hat{x})\\
	&\leq-\kappa(\mathcal{G}(x,\hat{x}))+\gamma_{int}(\Vert w- \hat{w}\Vert )+\hat{\gamma}(\eta)
	\end{align*}
	for any $\hat{x}_{d}\in\hat{f}(\hat{x},\hat{u},\hat{w})$.
	Using the previous inequality and by following a similar argument as the one in the proof of Theorem 1 in \cite{arxiv}, one obtains 
	\begin{align*}
	\mathcal{G}&(f(x,\mathcal{H}(x)+\hat{u},w),\hat{x}_{d})\\
	&\leq\max\{\tilde{\kappa}(\mathcal{G}(x,\hat{x})),\tilde{\gamma}_{int}(\Vert w- \hat{w}\Vert ),\tilde{\gamma}(\eta)\},
	\end{align*}
	where $\tilde{\kappa}=\mathcal{I}_d-(\mathcal{I}_d-\psi)\circ\hat{\kappa}$, $\tilde{\gamma}_{int}=(\mathcal{I}_d+\lambda)\circ\hat{\kappa}^{-1}\circ\psi^{-1}\circ\chi\circ\gamma_{int}$, $\tilde{\gamma}=(\mathcal{I}_d+\lambda^{-1})\circ\hat{\kappa}^{-1}\circ\psi^{-1}\circ\chi\circ(\chi-\mathcal{I}_d)^{-1}\circ\hat{\gamma}$,
	where $\lambda,\chi,\psi,\hat{\kappa}$ are some arbitrarily chosen $\mathcal{K}_{\infty}$ functions with $\mathcal{I}_d-\psi\in\mathcal{K}_{\infty},~\chi-\mathcal{I}_d\in\mathcal{K}_{\infty},~ \mathcal{I}_d-\hat{\kappa}\in\mathcal{K}_{\infty}$ and $\hat{\kappa}\leq\kappa $. Hence, inequality \eqref{e:SFC22} is satisfied 
	with $u=\mathcal{H}(x)+\hat{u}$, $\sigma=\tilde{\kappa}$, $\rho_{int}=\tilde{\gamma}_{int}$, $\rho_{ext}(s)= 0~ \forall s \in \R_{\ge0}$, $\varepsilon=\tilde{\gamma}(\eta)$, and, hence, $\mathcal{G}$ is an alternating simulation function from $\hat \Sigma$ to $\Sigma$.
	Similarly, we can also show that $\mathcal{G}$ is an alternating simulation function from $\Sigma$ to $\hat\Sigma$. In particular, by the definition of $\mathbb{\hat{U}}$, for any $u=\mathcal{H}(x)+\tilde{u}\in\mathbb{{U}}$ there always exists $\hat{u}\in\mathbb{\hat{U}}$ such that $\gamma_{ext}(\Vert \tilde{u}-\hat{u}\Vert)\leq \gamma_{ext}(\mu)$ which results in $\varepsilon=(\mathcal{I}_d+\lambda^{-1})\circ\hat{\kappa}^{-1}\circ\psi^{-1}\circ\chi\circ(\chi-\mathcal{I}_d)^{-1}\left(\gamma_{ext}(\mu)+\hat{\gamma}(\eta)\right)$. Other terms in the alternating simulation function $\mathcal{G}$ are the same as the first part of the proof.    		
\end{proof}	
\begin{remark}
	Observe that if $\gamma_{int}$ and $\hat{\gamma}$ are linear functions in the previous theorem, $\tilde{\gamma}_{int}$ and $\tilde{\gamma}$ reduce to $\tilde{\gamma}_{int}=(\mathcal{I}_d+\lambda)\circ\hat{\kappa}^{-1}\circ\psi^{-1}\circ\gamma_{int}$ and $\tilde{\gamma}=(\mathcal{I}_d+\lambda^{-1})\circ\hat{\kappa}^{-1}\circ\psi^{-1}\circ\hat{\gamma}$, respectively.
\end{remark}
\vspace{-0.75cm}
\begin{remark}		
		Although the choices of $\mathcal{K}_{\infty}$ functions $\lambda,\chi,\psi$, and $\hat{\kappa}$ in the previous theorem mainly depend on the dynamic of the given control systems, we provide a general guideline on choosing those functions as follows: $(i)$ In order to reduce the undesirable effect of the inverse of $\hat{\kappa}$ and $\psi$ in satisfying the small-gain condition in \eqref{SGC}, or in computing the value of the overall approximation error in \eqref{er}, one should choose those functions to behave very close to the identity function; $(ii)$ Regarding $\lambda$ and $\chi$, one should choose those functions such that the small gain condition in \eqref{SGC} is possibly satisfied, and then compute the overall approximation error in \eqref{er}. If the computed error is acceptable by the user, no further action is required; otherwise one should start slightly modifying those functions until a smaller error is achieved while ensuring that the small gain condition is not violated. For example, one can scale the $\mathcal{K}_{\infty}$ function $\lambda$ by a linear function $\beta(s)=cs \in \mathcal{K}_{\infty}$ , $\forall s\in\R_{\ge0}, c>1$, and then, using $\beta\circ\lambda$ instead of $\lambda$, start increasing the value of $c$ until a smaller error is obtained. Same procedure can be simultaneously applied to the $\mathcal{K}_{\infty}$ function $\chi$. It may be the case that the desired error is not achievable with the chosen $\lambda$ and $\chi$, then one should start over and choose different $\lambda$ and $\chi$ and go through similar procedure again.	
\end{remark}
Remark that condition \eqref{eq:TI} is not restrictive at all provided that one is interested to work on a compact subset of $\mathbb{X}$. We refer the interested readers to the explanation provided after equation (V.2) in \cite{zamani2014symbolic} on how to compute such function $\hat\gamma$. 

Now we provide similar results as in the first part of this section but tailored
to linear control systems which are computationally much more efficient. 
\subsection{Discrete-Time Linear Control Systems}\label{1:B}
The class of discrete-time linear control systems, considered in this subsection, is given by
\begin{align}\label{e:lin:sys}\Sigma:\left\{
\begin{array}[\relax]{rl}
x(k+1)&=Ax(k)+Bu(k)+Dw(k),\\
y(k)&=Cx(k),
\end{array}\right.
\end{align}
where $A\in\R^{n\times n}$, $B\in\R^{n\times m}$, $D\in\R^{n\times p}$, $C\in\R^{q\times n}$.
We use the tuple $\Sigma=(A,B,C,D)$
to refer to the class of control systems of the form~\eqref{e:lin:sys}. Remark that the incremental input-to-state stabilizability assumption in Definition \ref{def:sym} boils down in the linear case to the following assumption.
\begin{assumption}\label{assumption1}
	Let $\Sigma=(A,B,C,D)$. Assume that there exist matrices $Z\succ0$ and $K$ of appropriate dimensions such that the matrix inequality
	\begin{align}\label{e:lmi}
	(1+2\theta)(A+BK)^TZ(A+BK)\preceq \kappa_c Z	
	\end{align} holds for some constants $0<\kappa_c<1$, and $\theta\in\R_{>0}$.	
\end{assumption}
Note that condition \eqref{e:lmi} is nothing more than pair $(A,B)$ being stabilizable \cite{primer}. 
\begin{remark}
	Given constants $\kappa_c$ and $\theta$, one can easily see that inequality \eqref{e:lmi} is not jointly convex on decision variables $Z$ and $K$ and, hence, not amenable to existing semidefinite tools for linear matrix inequalities (LMI). However, using Schur complement, one can easily transform inequality \eqref{e:lmi} to the following LMI over decision variables $Q$ and $M$:
	\begin{align*}
	\begin{bmatrix}-\kappa_c Q & QA^T+M^TB^T  \\ AQ+BM & -(1+2\theta)Q\end{bmatrix}\preceq0,~Q\succ0,
	\end{align*}
	where $Q=Z^{-1}$ and $M=KQ$.
\end{remark}
Now, Theorem \ref{thm:2} reduces to the following one for linear systems. 
\begin{theorem}\label{Thm_3a}
	Consider $\Sigma=(A,B,C,D)$ and the finite abstraction $ \hat{\Sigma}$ constructed as in Definition \ref{def:sym}. Suppose Assumption \ref{assumption1} holds. Then, function 
	\begin{align}\label{e:lin:sf}
	V(x,\hat x)=\sqrt{(x-\hat x)^TZ(x-\hat x)},
	\end{align}
	is an alternating simulation function from $\hat \Sigma$ to $\Sigma$ and from $\Sigma$ to $\hat\Sigma$.
\end{theorem}
\begin{proof}{Proof.}
	First, we show that condition \eqref{e:SFC11} holds. Since $C = \hat C$, we have
	\begin{align*}
	\Vert Cx-\hat C\hat x\Vert\leq\sqrt{n\lambda_{\max}(C^TC)}\Vert x- \hat x\Vert,
	\end{align*}
	and similarity
	\begin{align*}
	\sqrt{\lambda_{\min}(Z)}\Vert x- \hat x\Vert\leq\sqrt{(x-\hat x)^TZ(x-\hat x)}.
	\end{align*}
	It can be readily verified that \eqref{e:SFC11} holds for $V$ defined in~\eqref{e:lin:sf} with $\alpha(s)=\sqrt{\frac{\lambda_{\min}(Z)}{n\lambda_{\max}(C^TC)}}s$ for any $s\in\R_{\geq0}$. 
	We continue to show that \eqref{e:SFC22} holds as well. Let $x$, $\hat x$, $\hat u$, and $\hat w$ be given, and choose $u$ as $	u\Let K(x-\hat x)+\hat u$. Let $x_d=Ax+Bu+Dw$, and $\hat{x}_d$ be defined as in Definition \ref{def:sym}. Define $F\Let  A\hat x +B\hat u+ D\hat w -\hat{x}_d$, and $\hat{\kappa}_c\Let 1-\sqrt{\kappa_c}$.
	\begin{figure*}[h]
		\rule{\textwidth}{0.4pt}
		\begin{small}
			\begin{align}\notag
			V(x_d,\hat{x}_{d})&=(Ax+Bu+Dw-(A\hat x+ B\hat u + D\hat w)+(A\hat x+ B\hat u + D\hat w)-\hat{x}_d)^TZ\\\notag
			&~~~~~(Ax+Bu+Dw-(A\hat x+ B\hat u + D\hat w)+(A\hat x+ B\hat u + D\hat w)-\hat{x}_d)^{\frac{1}{2}}\\\notag
			&=\big((x-\hat x)^T(A+BK)^TZ(A+BK)(x-\hat x)+(w-\hat w)^T D^TZ D(w-\hat w)+2 (w-\hat w)^T D^TZF\\\notag
			&~~~~~+2 (x-\hat x)^T(A+BK)^TZD(w-\hat w)+2(x-\hat x)^T(A+BK)^TZF+F^T Z F\big)^{\frac{1}{2}}\\\notag
			&\leq \big((x-\hat x)^T(A+BK)^TZ(A+BK)(x-\hat x)+(w-\hat w)^T D^TZ D(w-\hat w)+2 \Vert(w-\hat w)^T D^T\sqrt{Z}\Vert_2 \Vert \sqrt{Z}F\Vert_2\\\notag
			&~~~~~+2\Vert(x-\hat x)^T(A+BK)^T\sqrt{Z}\Vert_2 \Vert \sqrt{Z}F\Vert_2+2 \Vert(x-\hat x)^T(A+BK)^T\sqrt{Z}\Vert_2\Vert
			\sqrt{Z}D(w-\hat w)\Vert_2\\\notag&~~~~~+n\lambda_{\max}(Z)\eta^2\big)^{\frac{1}{2}} \\\notag
			&\leq \big((x-\hat x)^T(A+BK)^TZ(A+BK)(x-\hat x)+2\theta\Vert(x-\hat x)^T(A+BK)^T\sqrt{Z}\Vert^2_2+(w-\hat w)^T D^TZ D(w-\hat w)\\\notag
			&~~~~~+\frac{\Vert(w-\hat w)^T D^T\sqrt{Z}\Vert_2^2}{\theta}+2\frac{\Vert \sqrt{Z}F\Vert^2_2}{\theta}+ \theta\Vert(w-\hat w)^T D^T\sqrt{Z}\Vert^2_2+n\lambda_{\max}(Z)\eta^2\big)^{\frac{1}{2}}\\\label{eq:lin}
			&\leq \bigg((1+2\theta)(x-\hat x)^T(A+BK)^TZ(A+BK)(x-\hat x)+\frac{(1+\theta+\theta^2)(w-\hat w)^T D^TZ D(w-\hat w)}{\theta}\\\notag
			&~~~~~+\frac{n(2+\theta)\lambda_{\max}(Z)\eta^2}{\theta}\bigg)^{\frac{1}{2}}\\\notag
			&\leq\sqrt{\kappa_c}V(x,\hat{x})+\sqrt{\frac{1+\theta+\theta^2}{\theta}}\Vert \sqrt{Z}D\Vert_2 \Vert w-\hat w\Vert_2 +\sqrt{\frac{n(2+\theta)\lambda_{\max}(Z)}{\theta}}\eta\\\notag
			&\leq\sqrt{\kappa_c}V(x,\hat{x})+\sqrt{p\frac{1+\theta+\theta^2}{\theta}}\Vert \sqrt{Z}D\Vert_2 \Vert w-\hat w\Vert +\sqrt{\frac{n(2+\theta)\lambda_{\max}(Z)}{\theta}}\eta\\\notag
			&\leq(1-\hat{\kappa}_c)V(x,\hat{x})+\sqrt{p\frac{1+\theta+\theta^2}{\theta}}\Vert \sqrt{Z}D\Vert_2 \Vert w-\hat w\Vert +\sqrt{\frac{n(2+\theta)\lambda_{\max}(Z)}{\theta}}\eta.
			\end{align}
		\end{small}
		\rule{\textwidth}{0.4pt}
		\begin{small}
			\begin{align}\label{rlmi}
			V(x_d,\hat{x}_{d})
			\!\leq\! \max\Bigg\{\tilde{\kappa}\big((x-\hat x)^TZ(x-\hat x)\big)^{\frac{1}{2}},\frac{(1+\delta_c)}{\hat{\kappa}_c \psi_c}\sqrt{p\frac{(1\!+\!\theta\!+\!\theta^2)}{\theta}}\Vert \sqrt{Z}D\Vert_2 \Vert w-\hat w\Vert,\frac{(1+1/\delta_c)}{\hat{\kappa}_c \psi_c}\sqrt{\frac{n(2+\theta)\lambda_{\max}(Z)}{\theta}}\eta\Bigg\},	
			\end{align}
		\end{small}
		\rule{\textwidth}{0.4pt}
	\end{figure*}
	By following a similar argument as the one in the proof of Theorem 1 in \cite{arxiv}, one gets \eqref{rlmi}
	where $\tilde{\kappa}=(1-\hat{\kappa}_c(1-\psi_c))$, satisfying \eqref{e:SFC22} with $\sigma(s)=\tilde{\kappa}s$, $\rho_{ext}(s)= 0$,  $\rho_{int}(s)=\frac{(1+\delta_c)}{\hat{\kappa}_c \psi_c}\sqrt{p\frac{(1+\theta+\theta^2)}{\theta}}\Vert \sqrt{Z}D\Vert_2 s$, $\forall s\in\R_{\ge0}$, $\varepsilon=\frac{(1+1/\delta_c)}{\hat{\kappa}_c \psi_c}\sqrt{\frac{n(2+\theta)\lambda_{\max}(Z)}{\theta}}\eta$, where $\psi_c$ and $\delta_c$ can be chosen arbitrarily such that $0<\psi_c<1$ and $\delta_c>0$. Hence, the proposed $V$ in \eqref{e:lin:sf} is an alternating simulation function from $\hat \Sigma$ to $\Sigma$. The rest of the proof follows similar argument. In particular, by the definition of $\mathbb{\hat{U}}$, for any $u=K(x-\hat x)+\tilde{u}\in\mathbb{{U}}$ there always exists $\hat{u}\in\mathbb{\hat{U}}$ such that $\Vert B\Vert\Vert \tilde{u}-\hat{u}\Vert\leq \Vert B\Vert\mu$ which results in $\varepsilon=\frac{(1+1/\delta_c)}{\hat{\kappa}_c \psi_c}\sqrt{\frac{n(2+\theta)\lambda_{\max}(Z)}{\theta}}(\Vert B\Vert\mu+\eta)$. Other terms are the same as before.
\end{proof}
\section{Case Study}\label{1:V}
In this section we provide two case studies to illustrate our results and show their effectiveness in comparison with the existing compositional results in \cite{arxiv}. 
We first apply our results to the temperature regulation in a circular building by constructing compositionally a finite abstraction
of a network containing $n\geq 3$ rooms, each equipped with a heater. Then we
apply the proposed techniques to a fully connected network to show its
applicability to strongly connected networks as well. The construction of symbolic models and controllers are performed using tool SCOTS \cite{Rungger} on a PC with Intel i7@3.4GHz CPU and 16 GB of RAM.
\subsection{Room Temperature Control}\label{1}
The evolution of the temperature $\mathbf{T}$ of all rooms are described by the interconnected discrete-time model:
\begin{align*}
\Sigma:\left\{
\begin{array}{rl}
\mathbf{T}(k+1)&=A\mathbf{T}(k)+ \beta T_{E}+\mu T_{h}\nu(k),\\
\mathbf{y}(k)&=\mathbf{T}(k),
\end{array}\right.
\end{align*}
adapted from \cite{meyer}, where $A\in\R^{n\times n}$ is a matrix with elements $\{A\}_{ii}=(1-2\alpha-\beta-\mu\nu_{i}(k))$, $\{A\}_{i(i+1)}=\{A\}_{(i+1)i}=\{A\}_{1n}=\{A\}_{n1}=\alpha$, $\forall i\in [1;n-1]$, and all other elements are identically zero, $\mathbf{T}(k)=[\mathbf{T}_1(k);\ldots;\mathbf{T}_n(k)]$,  $\nu(k)=[\nu_1(k);\ldots;\nu_n(k)]$, $T_E=[T_{e1};\ldots;T_{en}]$, where $\nu_i(k)$, $\forall i\in[1;n]$, are taking values in $[0,0.6]$. The other parameters are as follow: $\forall i\in[1;n]$, $T_{ei}=-1\,^\circ C$ is the outside temperature, $T_h\!=\!50\,^\circ C$ is the heater temperature, and the conduction factors are given by $\alpha\!=\!0.45$, $\beta\!=\!0.045$, and $\mu\!=\!0.09$.\\
Now, by introducing $\Sigma_i$ described by
\begin{align*}
\Sigma_i:\left\{
\begin{array}{rl}
\mathbf{T}_i(k+1)&=a\mathbf{T}_i(k)+d\omega_i(k)+\beta T_{ei} +\mu T_h \nu_i(k),\\
\mathbf{y}_{i}(k)&=\mathbf{T}_i(k),%
\end{array}\right.
\end{align*}
one can readily verify that $\Sigma=\mathcal{I}_{0_n}(\Sigma_1,\ldots,\Sigma_n)$, where $a=1-2\alpha-\beta-\mu\nu_{i}(k) $, $d=[\alpha;\alpha]^T$, and $\omega_i(k)=[\mathbf{y}_{i-1}(k);\mathbf{y}_{i+1}(k)]$ (with $\mathbf{y}_{0}=\mathbf{y}_{n}$ and $\mathbf{y}_{n+1}=\mathbf{y}_{1}$).	
Note that for any $i\in[1;n]$, conditions \eqref{eq:ISTFC1} and \eqref{eq:ISTFC2} are satisfied with $\mathcal{G}_i(T_i,\hat{T}_i)=\Vert T_i-\hat T_i\Vert$, $\mathcal{H}_i\equiv0$, $\underline{\alpha}_{i}(s)=\overline{\alpha}_{i}(s)=s$, $\kappa_i(s)= (1-a)s$, $\gamma_{iint}(s)=\alpha s$, and $\gamma_{iext}\equiv0$. Furthermore, \eqref{eq:TI} is satisfied with $\hat{\gamma}=\mathcal{I}_d$. Consequently, $\mathcal{G}_i(T_i,\hat{T}_i)=\Vert T_i-\hat T_i\Vert$ is an alternating simulation function from $\hat\Sigma_i$,  constructed as in Definition \ref{def:sym}, to $\Sigma_i$.

Let, $\forall i\in[1;n]$, the $\mathcal{K}_{\infty}$ functions $\lambda_i$, $\psi_i$, and $\hat{\kappa}_i$ in the proof of Theorem \ref{thm:2} be as follows: $\lambda_i=\mathcal{I}_d$, $\psi_i(s)=0.99s$, $\hat{\kappa}_i=\kappa_i$. Since we have $\gamma_{ij}(s)<\mathcal{I}_d$, $\forall i,j\in[1;n],~ i\neq j$ and for any $n\geq 3$, the small-gain condition \eqref{SGC} is satisfied without any restriction on the number of rooms.
Using the results in Theorem \ref{thm:3} with $\delta^{-1}_{i}=\mathcal{I}_d,~ \forall i\in [1;n]$, one can verify that $V(T,\hat{T})=\max_{i}\{\Vert T_i-\hat T_i\Vert\}$ is an alternating simulation function from $\hat{\Sigma}=\mathcal{I}_{0_n}(\hat\Sigma_1,\ldots,\hat\Sigma_n)$ to $\Sigma$ satisfying conditions \eqref{e:SFC11a} and \eqref{e:SFC22b} with $\tilde{\sigma}(s)=\max\left\{(1-(1-a)10^{-2})s,\frac{2.02\alpha}{1-a} s\right\}$, $\tilde{\alpha}(s)=s$, $\tilde{\rho}_{ext}(s)=0$ $\forall s\in\R_{\ge0}$, $\tilde{\varepsilon}=\max_{i}\left\{\frac{2.02\eta_i}{1-a}\right\}$, $\forall i\in[1;N]$, where $\eta_i$ is the state set quantization parameter of abstraction $\hat\Sigma_i$. 

Remark that, to have a fair comparison with the compositional technique proposed in \cite{arxiv}, we have assumed that $\mathbb{\hat{Y}}_{ji}= \mathbb{\hat W}_{ij}$, i.e. $\varpi_{ij}=0$, $\forall i,j \in[1;n],~i\neq j$.  
For the fair comparison, we compute error $\hat\varepsilon$ in the $\hat\varepsilon$-approximate alternating simulation relation as in \eqref{er} based on the dissipativity approach in \cite{arxiv} and the small-gain approach here. This error represents the mismatch between the output behavior of the concrete interconnected system $\Sigma$ and that of its finite abstraction $\hat{\Sigma}$. We evaluate $\hat \varepsilon$ for different number of subsystems $n$ and different values of the state set quantization parameters $\eta_i$ for abstractions $\hat\Sigma_i~\forall i\in [1;n]$ as in Figure \ref{tc}.
As shown, the small-gain approach results in less mismatch errors than those obtained using the dissipativity based approach in \cite{arxiv}. 
The reason is that the error in \eqref{er} is computed based on the maximum of the errors between concrete subsystems and their finite abstractions instead of being a linear combination of them which is the case in \cite{arxiv}. Hence, by increasing the number of subsystems, our error does not change here whereas the error computed by the dissipativity based approach in \cite{arxiv} will increase as shown in Figure \ref{tc}.

\begin{figure}[thp]
	\begin{center}
		\includegraphics[height=7.0cm]{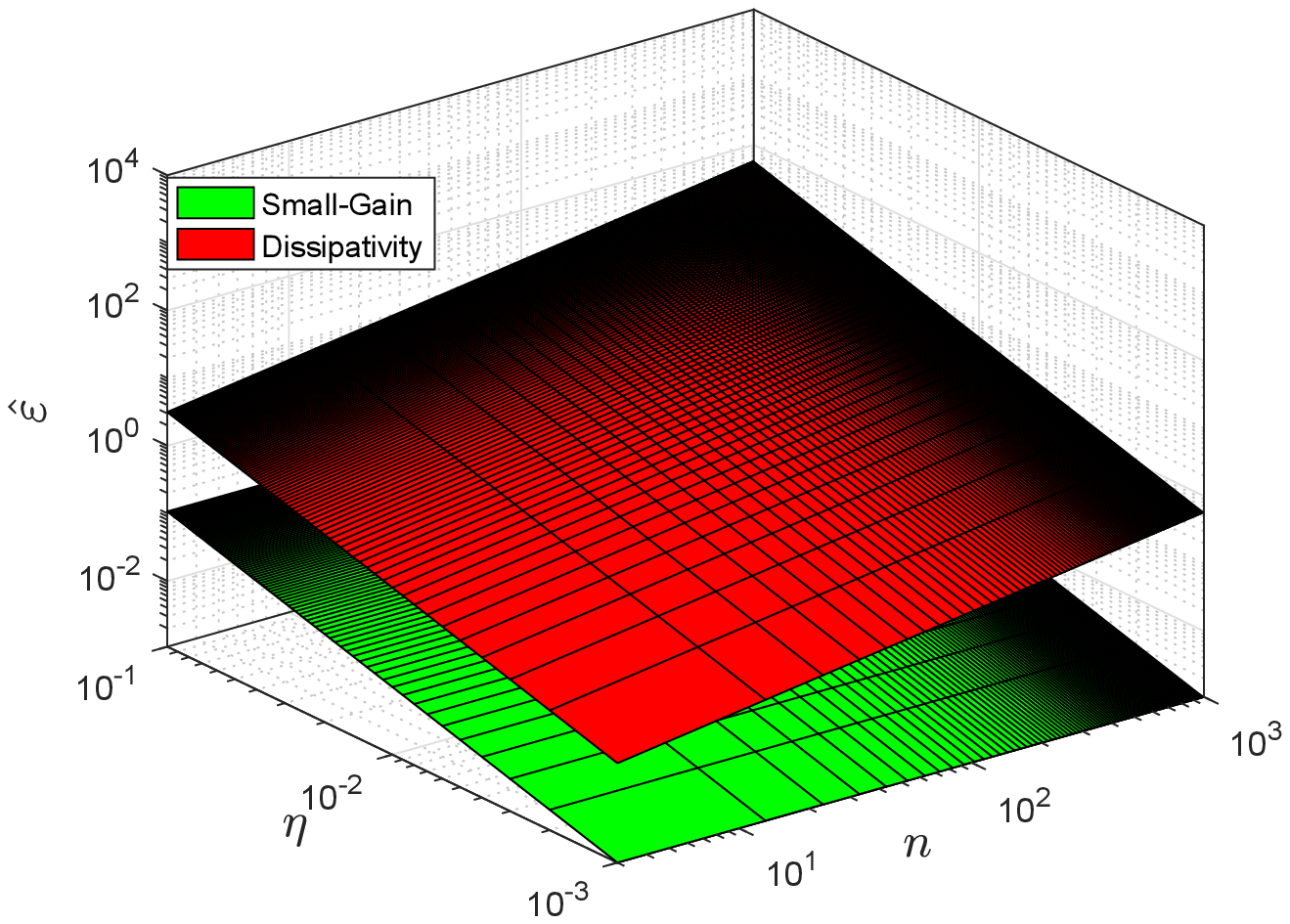}
		\caption{Temperature control: Comparison of errors
			in \eqref{er} resulted from our approach based on small-gain
			condition with those resulted from the approach in \cite{arxiv} based on
			dissipativity-type condition for different values of $n\geq3$ and $\eta_i$.}
		\label{tc}
	\end{center}
\end{figure}

Now, we synthesize a controller for $\Sigma$ via abstractions $\hat{\Sigma}_i$ such that the temperature of each room is maintained in
the comfort zone $\mathcal{S=}[19,~21]$. The idea here is to design local controllers
for abstractions $\hat{\Sigma}_i$, and then refine them to concrete subsystems $\Sigma_i$. To do so, the local controllers are synthesized
while assuming that the other subsystems meet their safety specifications. This approach, called assume-guarantee reasoning, allows for the compositional synthesis of controllers as well. The computation times for constructing abstractions and synthesizing controllers for $\Sigma_i$ are $0.048$s and $0.001$s, respectively.  Figure \ref{st} shows the state trajectories of the closed-loop system $\Sigma$, consisting of $1000$ rooms, under control inputs $u_i$ with the state and input quantization parameters $\eta_i=0.01$ and $\mu_{i}=0.01$, $\forall i\in [1;1000]$, respectively.
\begin{figure}
	\begin{center}
		\includegraphics[height=7.0cm]{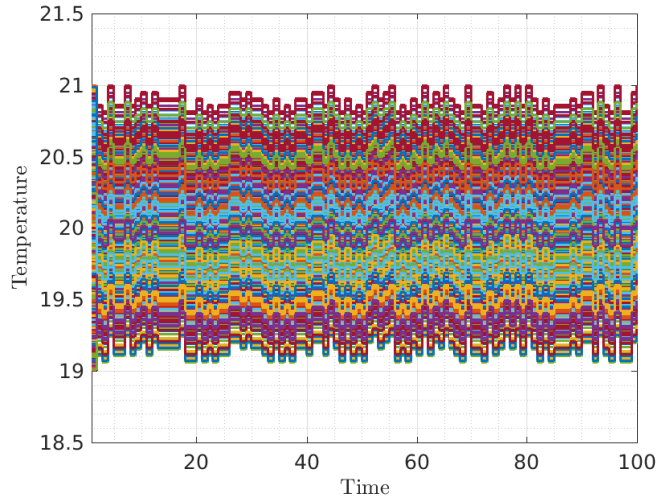}
		\caption{State trajectories of the closed-loop system $\Sigma$ consisting of $1000$ rooms.}
		\label{st}
	\end{center}
\end{figure}
\subsection{Fully Connected Network}\label{1}
In order to show the applicability of our approach to strongly
connected networks, we consider a nonlinear control system $\Sigma$ described by 
\begin{align*}
\Sigma:\left\{
\begin{array}{rl}
\mathbf{x}(k+1)&=A\mathbf{x}(k)+\varphi(x)+\nu(k),\\
\mathbf{y}(k)&=\mathbf{x}(k),
\end{array}\right.
\end{align*}
where $A=I_n-\tau L$ for some Laplacian matrix $L \in\R^{n\times n}$ of an undirected graph \cite{Godsil}, and constant $0<\tau <1/\Delta$, where $\Delta$ is the maximum degree of the
graph \cite{Godsil}. Moreover $\mathbf{x}(k)=[\mathbf{x}_1(k);\ldots;\mathbf{x}_n(k)]$, $\nu(k)=[\nu_1(k);\ldots;\nu_n(k)]$, and $\varphi(x)=[\varphi_1(x_1);\ldots;\varphi_n(x_n)]$, where $\varphi_i(x_i)=sin(x_i), \forall i\in[1;n]$. Assume $L$ is the Laplacian matrix of a complete graph:
\begin{align*}
\begin{array}{rl}
L=\begin{bmatrix}n-1 & -1 & \cdots & \cdots & -1 \\  -1 & n-1 & -1 & \cdots & -1 \\ -1 & -1 & n-1 & \cdots & -1 \\ \vdots &  & \ddots & \ddots & \vdots \\ -1 & \cdots & \cdots & -1 & n-1\end{bmatrix}.
\end{array}
\end{align*}
Now, by introducing $\Sigma_i$ described by
\begin{align*}
\Sigma_i:\left\{
\begin{array}{rl}
\mathbf{x}_i(k+1)&=a_i\mathbf{x}_i(k)+\varphi_i(x_i)+d_i\omega_i(k)+\nu_i(k),\\
\mathbf{y}_{i}(k)&=\mathbf{x}_i(k),
\end{array}\right.
\end{align*}
where $a_i\!=\!\{A\}_{ii}$, $\omega_i(k)\!=\![\mathbf{y}_{i1};\!\ldots\!;\mathbf{y}_{i(i-1)};\mathbf{y}_{i(i+1)};\!\ldots\!;\mathbf{y}_{in}]$, $d_i=[\{A\}_{i1};\ldots;\{A\}_{i(i-1)};\{A\}_{i(i+1)};\ldots;\{A\}_{in}]^T$, one can readily verify that $\Sigma=\mathcal{I}_{0_n}(\Sigma_1,\ldots,\Sigma_n)$.
Clearly, for any $i\in[1;n]$, conditions \eqref{eq:ISTFC1} and \eqref{eq:ISTFC2} are satisfied with $\mathcal{G}_i(x_i,\hat{x}_i)=\Vert x_i-\hat x_i\Vert$, $\mathcal{H}_i(x_i)=-c_ix_i$, where $\frac{a_i+1}{2}<c_i<a_i+1$, $\underline{\alpha}_{i}(s)=\overline{\alpha}_{i}(s)=s$, $\kappa_i(s)= \left(1-(1+a_i-c_i)\right)s$, $\gamma_{iint}(s)=\Vert d_i \Vert s$, and $\gamma_{iext}(s)=0$, $\forall s\in\R_{\ge0}$. Note that \eqref{eq:TI} is satisfied with $\hat{\gamma}=\mathcal{I}_d$. Consequently, $\mathcal{G}_i(x_i,\hat{x}_i)=\Vert x_i-\hat x_i\Vert$ is an alternating simulation function from $\hat\Sigma_i$,  constructed as in Definition \ref{def:sym}, to $\Sigma_i$.

Fix $\tau=\frac{0.1}{\Delta}=\frac{0.1}{n-1}$, and let, $\forall i\in[1;n]$, the $\mathcal{K}_{\infty}$ functions $\lambda_i$, $\psi_i$, and $\hat{\kappa}_i$ in the proof of Theorem \ref{thm:2} be as follows: $\lambda_i=\mathcal{I}_d$, $\psi_i(s)=0.99s$, $\hat{\kappa}_i=\kappa_i$. Since we have $\gamma_{ij}(s)<\mathcal{I}_d$, $\forall i,j\in[1;n],~ i\neq j$, the small-gain condition \eqref{SGC} is satisfied without any restriction on the number of subsystems.
Using the results in Theorem \ref{thm:3} with $\delta^{-1}_{i}=\mathcal{I}_d~ \forall i\in [1;n]$, one can verify that $V(x,\hat{x})=\max_{i}\{\Vert x_i-\hat x_i\Vert\}$ is an alternating simulation function from $\hat{\Sigma}={\mathcal{I}}_{0_n}(\hat\Sigma_1,\ldots,\hat\Sigma_n)$ to $\Sigma$
satisfying conditions \eqref{e:SFC11a} and \eqref{e:SFC22b} with $\tilde{\alpha}(s)=s$, $\tilde{\rho}_{ext}(s)=0$, $\forall s\in\R_{\ge0}$, $\tilde{\varepsilon}=\max_{i}\left\{\frac{2.02\eta_i}{1-(1+a_i-c_i)}\right\}$, $\tilde{\sigma}(s)\!=\!\max\left\{\!\max\limits_{i}\!\left\{\!\left(1\!-\!\frac{(1-(1+a_i-c_i))}{10^{2}}\right)s\right\}\!,\max\limits_{i}\left\{\!\frac{2.02\Vert d_i \Vert}{1-(1+a_i-c_i)} s\!\right\}\!\right\}\!,$ where $\eta_i$ is the state set quantization parameter of abstraction $\hat\Sigma_i$.

Similar to the previous case study, in order to compare our compositional technique to the one proposed in \cite{arxiv}, we have assumed that $\mathbb{\hat{Y}}_{ji}= \mathbb{\hat W}_{ij}$, i.e. $\varpi_{ij}=0$, $\forall i,j \in[1;n],i\neq j$.
A comparison of the error $\hat\varepsilon$ in \eqref{er} resulted from the dissipativity approach in \cite{arxiv} and the small-gain approach here is shown in Figure \ref{lap}. We compute $\hat \varepsilon$ for different number of subsystems $n$ and different values of the state set quantization parameters $\eta_i$ for abstractions $\hat\Sigma_i,~\forall i\in [1;n]$. 
Clearly, the small-gain approach results in less mismatch errors than those obtained using the dissipativity based approach in \cite{arxiv}. 

The computation time for constructing abstractions for $\Sigma_i$ is $0.9$s after fixing $n=1000$, $\eta_i=0.01$, $\mu_{i}=0.01$, $x_i\in[0,10]$, $\nu_i\in[0,1]$, $\forall i\in [1;n]$.
\begin{figure}
	\begin{center}
		\includegraphics[height= 7.0cm]{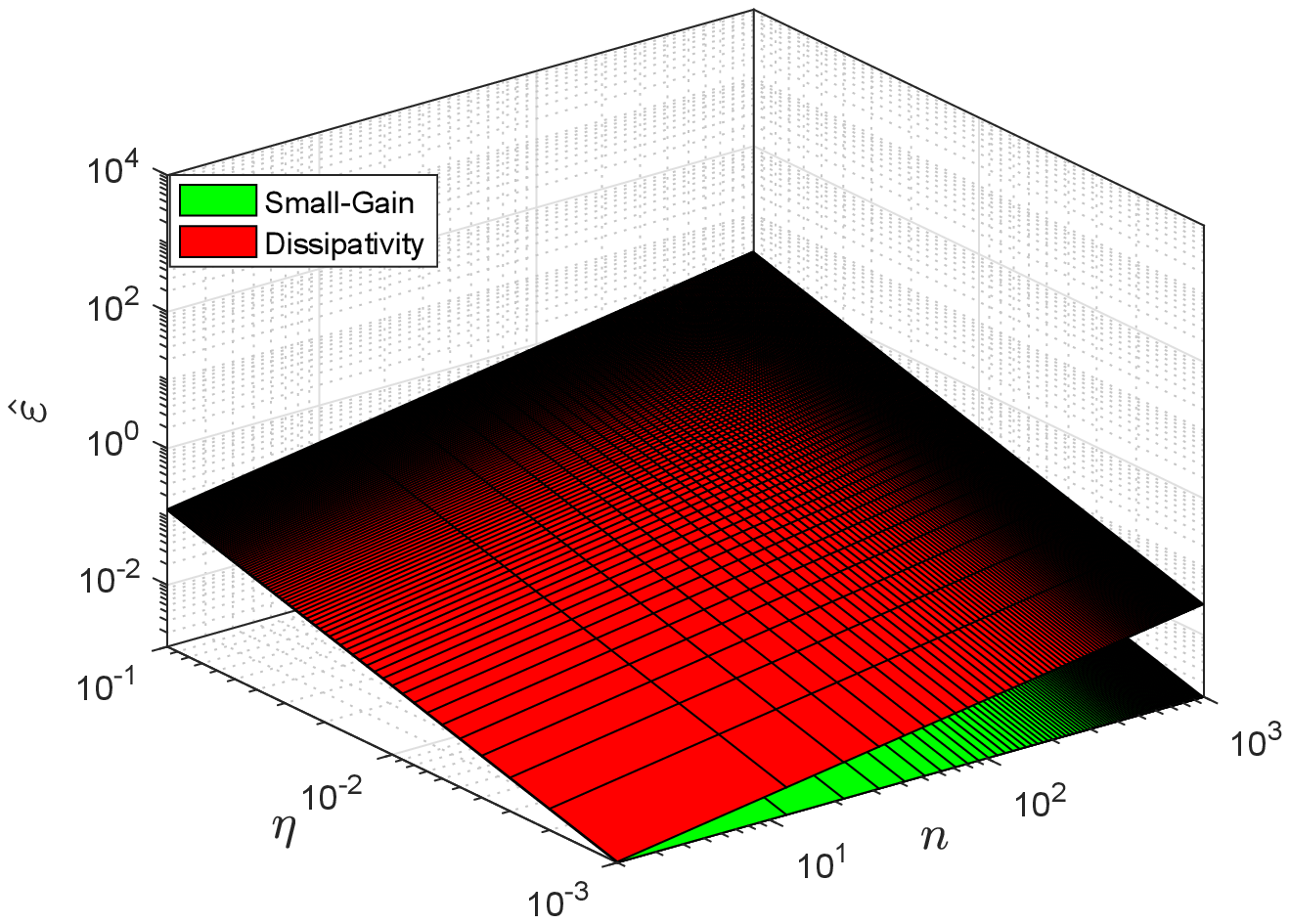}
		\caption{Fully connected network: Comparison of errors
			in \eqref{er} resulted from our approach based on small-gain
			condition with those resulted from the approach in \cite{arxiv} based on
			dissipativity-type condition for different values of $n\geq1$ and $\eta_i$.}
		\label{lap}
	\end{center}
\end{figure}
\section{Conclusion}
In this paper, we proposed a compositional framework for the construction of finite abstractions of interconnected discrete-time control systems. First, we used a notion of so-called alternating simulation functions in order to construct compositionally an overall alternating simulation function that is used to quantify the error between the output behavior of the overall interconnected concrete system and the one of its finite abstraction. 
Furthermore, we provided a technique to construct finite abstractions together with their corresponding alternating simulation functions for discrete-time control systems under incremental input-to-state stabilizability property.
Finally, we illustrated the proposed results by constructing finite abstractions of two networks of (linear and nonlinear) discrete-time control systems and their corresponding alternating simulation functions in a compositional fashion. We elucidated the effectiveness of our compositionality results in comparison with the existing ones using dissipativity-type reasoning.

\bibliographystyle{alpha}       
\bibliography{j1}           



\appendix
\end{document}